\documentclass[english,aps,prl,superscriptaddress,floatfix, notitlepage,reprint]{revtex4-1}
\usepackage[T1]{fontenc}
\usepackage[latin9]{inputenc}
\usepackage{mathrsfs}
\usepackage{amsmath}
\usepackage{amsthm}
\usepackage{amssymb}
\usepackage{graphicx}

\makeatletter
\theoremstyle{plain}
\newtheorem{thm}{\protect\theoremname}
\theoremstyle{plain}
\newtheorem{lem}[thm]{\protect\lemmaname}

\usepackage{braket}
\usepackage{txfonts} 
\usepackage[usenames,dvipsnames]{xcolor}
\usepackage[colorlinks=true,citecolor=Blue,linkcolor=RubineRed,urlcolor=Blue]{hyperref}
\date{\today}

\makeatother

\usepackage{babel}
\providecommand{\lemmaname}{Lemma}
\providecommand{\theoremname}{Theorem}

\begin{document}
\title{Variational principle for optimal quantum controls in quantum metrology}
\author{Jing Yang\href{https://orcid.org/0000-0002-3588-0832} {\includegraphics[scale=0.05]{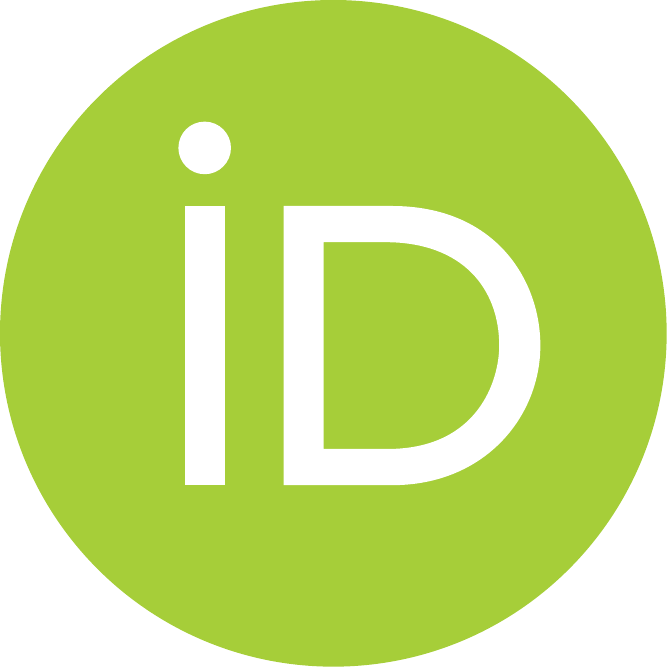}}}

\email{jing.yang@uni.lu}
\thanks{These two authors contribute equally}
\address{Department of Physics and Materials Science, University of Luxembourg,
L-1511 Luxembourg, Luxembourg}
\author{Shengshi Pang\href{https://orcid.org/0000-0002-6351-539X}{\includegraphics[scale=0.05]{orcidid.pdf}}}
\email{pangshsh@mail.sysu.edu.cn }
\thanks{These two authors contribute equally}
\affiliation{School of Physics, Sun Yat-Sen University, Guangzhou, Guangdong 510275,
China}
\author{Zekai Chen\href{https://orcid.org/0000-0001-8676-9731}{\includegraphics[scale=0.05]{orcidid.pdf}}}

\affiliation{Department of Physics and Astronomy, University of Rochester, Rochester, New York 14627, USA}
\author{Andrew N. Jordan \href{https://orcid.org/0000-0002-9646-7013}{\includegraphics[scale=0.05]{orcidid.pdf}}}
\affiliation{Institute for Quantum Studies, Chapman University, 1 University Drive, Orange, CA 92866, USA}
\affiliation{Department of Physics and Astronomy, University of Rochester, Rochester, New York 14627, USA}
\author{Adolfo del Campo\href{https://orcid.org/0000-0003-2219-2851}{\includegraphics[scale=0.05]{orcidid.pdf}}}
\email{adolfo.delcampo@uni.lu}
\address{Department of Physics and Materials Science, University of Luxembourg,
L-1511 Luxembourg, Luxembourg}
\address{Donostia International Physics Center, E-20018 San Sebasti\'an, Spain}

\begin{abstract}
We develop a variational principle to determine the quantum controls and
initial state which optimizes the quantum Fisher information,
the quantity characterizing the precision in quantum metrology.
When the set of available controls is limited, the exact optimal initial state
and the optimal controls are in general dependent on the probe time,
a feature missing in the unrestricted case.
Yet, for time-independent Hamiltonians with restricted
controls, the problem can be approximately reduced to the unconstrained
case via the Floquet engineering. In particular, we find for  magnetometry
with a time-independent spin chain containing three-body interactions,
even when the controls are restricted to one and two-body interaction,
that the Heisenberg scaling can still be approximately achieved. Our results
open the door to investigate quantum metrology under a limited set
of available controls, of relevance to  many-body quantum
metrology in realistic scenarios. 
\end{abstract}
\maketitle

\textit{Introduction}.  At the heart of quantum technology, quantum
metrology aims at improving the precision of parameter
estimation~\cite{helstrom1976quantum,holevo2011probabilistic,braunstein1994statistical,paris2009quantum,Giovannetti2011}.
Recent theoretical advances in quantum metrology~\cite{yuan2015optimal,yuan2016sequential,pang2017optimal,yang2017quantum,liu2017quantum,yang2021hamiltonian} have
led to the introduction of optimal control protocols that maximizes the quantum
Fisher information (QFI), the key quantity in characterizing the precision
in quantum metrology. The use of control protocols for enhanced parameter estimation has been extensively explored in many practical
physical setups~\cite{pang2017optimal,yang2017quantum,yang2021hamiltonian,gefen2019overcoming,naghiloo2017achieving}. 

The identification of the optimal control protocol for parameter estimation ~\cite{pang2017optimal,yang2017quantum} can be done exploiting the notion of adiabatic continuation of quantum states.  It resembles the engineering of shortcuts to adiabaticity (STA) by counterdiabatic driving \cite{demirplak2003adiabatic,demirplak2005assisted,demirplak2008onthe,berry2009transitionless,STArev19},  whereby the system Hamiltonian is supplemented by some controls that enforce adiabatic continuation along a prescribed trajectory. In the context of STA, the auxiliary controls enforce parallel transport along the eigenstates of the uncontrolled system Hamiltonian. 
By contrast, optimal controls for quantum metrology enforce adiabatic continuation of the eigenstates of an operator defined by  the parametric derivative
of the estimation Hamiltonian ~\cite{cabedo-olaya2020shortcuttoadiabaticitylike}. This is intuitive as such operator guarantees maximal distinguishability of states with a slight change of the estimated parameter. The connection of optimal control for quantum metrology and STA is not only fruitful in providing a geometric justification of the required controls. As we show in this work, it suggests solutions to common challenges.

In STA, the exact identification of the counterdiabatic control protocol is not feasible in many-body systems in which spectral properties of the system Hamiltonian cannot be found, e.g., in problems such as quantum optimization.
In addition, the set of available controls in a given experimental setup is often restricted. This generally precludes the implementation of the exact STA protocol. This is particularly  important in the many-particle systems when exact counterdiabatic controls involve non-local multiple-body interactions beyond one-body and pair-wise potentials which are hard to realize in practice \cite{delcampo12}. In addition, the use of variational methods provides then an alternative, by designing optimal protocols that are  realizable with a restricted set of controls \cite{carlini2007timeoptimal, carlini2006timeoptimal, rezakhani2009quantum, Opatrny14,Saberi14, sels2017minimizing,claeys2019floquetengineering,Chandarana21, wang2015quantum}. 
Likewise, the exact construction  of optimal protocols  in many-body quantum metrology~\cite{yang2021hamiltonian} (i) requires 
access to the spectrum of the parametric derivative of the many-body
Hamiltonian which is hardly accessible in general, (ii)  may require optimal controls that are nonlocal,
and hard to  implement  in the laboratory. 
Even in single-qubit metrology using a NV center, superconducting qubit, or
quantum dot, certain control operations may be hard to implement. It is thus required to develop a new formalism of optimal control for quantum metrology beyond the state of the art \cite{pang2017optimal}, by taking into account the ubiquitous limitations on the controls and circumventing the requirement to access the spectral properties of the Hamiltonian derivative. This is the problem solved in this Letter by  means of a novel variational approach that relies on a metrological action. We note that our work is  unrelated to other variational approaches  introduced in quantum metrology that do not involve quantum control~\cite{kaubruegger2021quantum,marciniak2021optimal}.



Specifically, we assume that the space of control Hamiltonians is spanned by a set of
local basis operators and introduce a metrological action, which includes contributions of the quantum Fisher information and the Schr\"odinger equation, which
is considered as a constraint to the optimization problem. We derive
the optimal control conditions and show that it reduces to the unrestricted
protocol in Ref.~\cite{pang2017optimal} if the basis operators generate
the full space of the Hermitian operators. Furthermore, we show that
when the control Hamiltonian is restricted, the \textit{exact} optimal
solution for the initial state and the control Hamiltonian depends
on the probe time. However, we find that for multiplicative time-independent
Hamiltonians, even with limited  controls, one may identify
 \textit{approximate }optimal\textit{ }controls, which reduce
to the unrestricted protocol and therefore become independent of the probe
time, using the high frequency expansion in Floquet engineering~\cite{rahav2003timeindependent,goldman2014periodically,bukov2015universal,eckardt2015highfrequency,goldman2015erratum,eckardt2017colloquium,chen2020textbackslash}.
In particular, we show in an example in magnetometry that one can
avoid local three-body interaction terms making use
only of local two-body control Hamiltonians.

\textit{Optimal initial states and controls through variational principle}.
Consider the  quantum  estimation of a parameter $\lambda$ in a general Hamiltonian $H_{\lambda}(t)$.
In this setting, one may find that the QFI may decrease
as the probe time increases~\cite{pang2014quantum}. This motivates the introduction of quantum
controls $H_{\text{c}}(t)$ to 
enhance the QFI ~\cite{yuan2015optimal,pang2017optimal}.
When the quantum controls are introduced, the unitary evolution $U(t)$ is generated by $H_{\text{tot}}(t)=H_{\lambda}(t)+H_{\text{c}}(t)$.
Let the initial time and fixed final probe time be $0$ and 
$t_{f}$, respectively. The generator for parameter estimation
is given by $G_{t_{f}}[U]=\int_{0}^{t_{f}}U^{\dagger}(\tau)\partial_{\lambda}H_{\lambda}(\tau)U(\tau)d\tau$~\cite{pang2017optimal}
and the quantum Fisher information $I$ is given by the variance of the generator~\cite{paris2009quantum}, i.e., $I[U]=\text{Var}(G_{t_{f}}[U])$.

The optimization of the QFI $I[U]$ over the initial state yields 
the optimal initial state $(\ket{\varphi_{+}(t_{f})}+\ket{\varphi_{-}(t_{f})})/\sqrt{2}$,
where $\ket{\varphi_{\pm}(t_{f})}$ are referred to as the maximum and
minimum eigenstates, as they are associated with the maximum and minimum eigenvalues $\mu_{\pm}(t_{f})$
of $G_{t_{f}}[U]$~\cite{yang2021SM}.  The maximum value
of the QFI over the initial states $\ket{\psi_{0}}$ is $\max_{\ket{\psi_{0}}}I=\|G_{t_{f}}[U]\|^{2}/4$,
where the norm of an operator is defined by the difference between
its maximum and minimum eigenvalues~\cite{braunstein1996generalized, boixo2007generalized, Giovannetti2011}.
Our next goal is to maximize the quantum Fisher information $I[U]$
over all possible unitary dynamics under the condition that $U$ and
$H_{\text{c}}$  satisfy the Schr\"odinger equation
$\text{i}\partial_{t}U(t)=H_{\text{tot}}(t)U(t)$ as a constraint. We further  require 
the control Hamiltonian $H_{\text{c}}$ to be spanned by a limited set of
available linearly-independent terms 
is $\{\mathcal{X}_{i}\}_{i=1}^{d_{\text{c}}}$. We denote $\mathcal{V}_{\text{c}}\equiv\text{span}\{\mathcal{X}_{i}\}_{i=1}^{d_{\text{c}}}$
and expand $H_{\text{c}}\in\mathcal{V}_{\text{c}}$ as $H_{\text{c}}(\tau)=\sum_{i=1}^{d_{\text{c}}}c_{i}(\tau)\mathcal{X}_{i}$.
This procedure amounts to removing the constraint on $H_{\text{c}}$, reducing
 the optimization over $H_{\text{c}}$ to the optimization
over $c_{i}(\tau)~$~\footnote{This is the approach pursued in e.g., the variational approach in
shortcut-to-adiabaticity~\cite{sels2017minimizing,claeys2019floquetengineering}.
Another way of handling the constraints on the control Hamiltonian
is by introducing the following constraints $f_{j}(H_{\text{c}}(\tau))=\text{Tr}\left\{ H_{\text{c}}(\tau)\mathcal{X}_{j}\right\} =0,\,j=d_{\text{c}}+1,\,\cdots N$
to disallow the terms $\mathcal{X}_{j}$, where $j=1,\,2\,\cdots d_{\text{c}}$.
This way of introducing the constraint is the one used in e.g. quantum
brachistochrone equation~\cite{carlini2007timeoptimal,carlini2006timeoptimal}.
However, in many-body quantum metrology, the number of disallowed
nonlocal operators is much more than the allowed local operators.
Therefore the second approach may introduce an intractable number of
constraints and we shall pursue the first approach of expanding $H_{\text{c}}$
in terms of basis operators in the main text.}.

We shall denote $I_{0}=\max_{U,\,H_{\text{c}},\,\ket{\psi_{0}}}I[U]$. In principle, $I_{0}$ can be computed through the variational principle by constructing an appropriate action. 
We note that $I[U]$ is quartic in $U$ since it is quadratic in $G_{t_f}[U]$, which is itself
quadratic in $U$. This makes the
variational calculus of $I[U]$ with respect to $U$ tedious. To facilitate the calculation,
we observe that
\begin{align}
I_{0}= & \max_{U,\,H_{\text{c}}}\max_{\ket{\varphi_{a,b}}}\left(\braket{\varphi_{a}\big|G_{t_{f}}[U]\big|\varphi_{a}}-\braket{\varphi_{b}\big|G_{t_{f}}[U]\big|\varphi_{b}}\right)^{2}\label{eq:I0}
\end{align}
under the constraint of Schr\"odinger equation and the condition that
$\ket{\varphi_{a,\,b}}$ are normalized. The introduction of two
more optimization variables $\ket{\varphi_{a,\,b}}$ allows us to
remove the square in Eq.~(\ref{eq:I0}) and
transform the original optimization problem to the following equivalent
one, $\max_{\ket{\varphi_{a,b}},\,U,\,H_{\text{c}}}S_{\text{I}}[\Delta\rho,\,U]$, 
under aforementioned constraints, with the "information action"
being defined as
\begin{align}
S_{\text{I}}[\Delta\rho,\,U] & \equiv\braket{\varphi_{a}\big|G_{t_{f}}[U]\big|\varphi_{a}}-\braket{\varphi_{b}\big|G_{t_{f}}[U]\big|\varphi_{b}}\nonumber \\
 & =\int_{0}^{t_{f}}\text{Tr}\left\{ \Delta\rho U^{\dagger}(\tau)\partial_{\lambda}H_{\lambda}(\tau)U(\tau)\right\} d\tau,
\end{align}
and $\Delta\rho\equiv\ket{\varphi_{a}}\bra{\varphi_{a}}-\ket{\varphi_{b}}\bra{\varphi_{b}}$.
The introduction of two additional auxiliary variables
$\ket{\varphi_{a,\,b}}$ effectively renders $S_{\text{I}}[U]$ quadratic in $U$, unlike
$I[U]$,  facilitating the variational calculus with
respect to $U$. Upon introducing the Lagrangian multipliers $\mu_{a,\,b}$
and $\Lambda(\tau)$, we obtain the following "metrological action"
\begin{align}
S_{\text{M}}(\ket{\varphi_{a}},\,\ket{\varphi_{b}},\,U,\,H_{\text{c}}) & \equiv S_{\text{I}}[\Delta\rho,\,U]+S_{\text{S}}[U,\,H_{\text{c}}]\nonumber \\
 & -\mu_{a}[\braket{\varphi_{a}\big|\varphi_{a}}-1]+\mu_{b}[\braket{\varphi_{b}\big|\varphi_{b}}-1]\label{eq:metrology-action}
\end{align}
where the "Schr\"odinger action" is defined as 
\begin{equation}
S_{\text{S}}[U,\,H_{\text{c}}]\equiv\int_{0}^{t_{f}}\text{Tr}\left\{ \Lambda(\tau)[\text{i}\dot{U}(\tau)U^{\dagger}(\tau)-H_{\lambda}(\tau)-H_{\text{c}}(\tau)]\right\} d\tau.
\end{equation}
We emphasize that $\ket{\varphi_{a,\,b}}$, 
$U$ and $H_{\text{c}}$ are independent variables. The optimization
over $\ket{\varphi_{a,\,b}}$ can be easily implemented by differentiation
with respect to them, which yields 
\begin{equation}
G_{t_{f}}[U]\ket{\varphi_{\alpha}}=\mu_{\alpha}\ket{\varphi_{\alpha}},\,\alpha=a,\,b.\label{eq:G-eigen}
\end{equation}
As shown in Sec.~I in~\cite{yang2021SM}), in order for $I_{0}$ to take the global maximum values over $\ket{\varphi_{a,\,b}}$,
$\mu_{a,\,b}$ and $\ket{\varphi_{a,\,b}}$ must be the maximum
and minimum eigenvalues and eigenvectors of $G_{t_{f}}[U]$, respectively.

Variation with respect to $H_{\text{c}}$ and $U$ gives the trace condition
$\text{Tr}\left\{ \Lambda(\tau)\mathcal{X}_{i}\right\} =0$ for $i\in[1,\,d_{\text{c}}]$
and the differential equation $\dot{\Lambda}(\tau)-\text{i}[\Lambda(\tau),\,H_{\text{tot}}(\tau)]+\text{i}[\Delta\rho(\tau),\,\partial_{\lambda}H_{\lambda}(\tau)]=0$,
with the final condition $\Lambda(t_{f})=0$, where $\Delta\rho(\tau) \equiv U(\tau)\Delta\rho U^{\dagger}(\tau)$. One can
solve for $\Lambda(\tau)$ the differential equation and substitute the result
into the trace condition to find 
\begin{equation}
\text{Tr}\left\{ \mathcal{X}_{i}\partial_{\lambda}[\Delta\rho(\tau)]\right\} =0,\,i\in[1,\,d_{\text{c}}].\label{eq:no-leakage}
\end{equation}
Eq.~\eqref{eq:G-eigen} and Eq.~\eqref{eq:no-leakage} are our central
results. The form of these equations in the parameter-independent rotating frame can be found in~\cite{yang2021SM}. They give the optimal initial state and optimal dynamics that maximize the QFI  when the quantum controls are restricted
to the subspace spanned by $\{\mathcal{X}_{i}\}$. In what follows,
we  discuss their implications and applications.

\textit{The unrestricted control and the general feature of exact
restricted controls.} As a first application of our results, let
us assume there is no restriction on the control Hamiltonians, that
is, $\{\mathcal{X}_{j}\}$ spans the whole space of traceless Hermitian
operators. We shall see how the Pang-Jordan protocol in Ref.~\cite{pang2017optimal}
is reproduced. In this case, Eq.~(\ref{eq:no-leakage}) is equivalent
to $\partial_{\lambda}[\Delta\rho(\tau)]=0$. Taking the time derivative
on both sides yields
\begin{equation}
\partial_{\tau}\partial_{\lambda}[\Delta\rho(\tau)]=-\text{i}[\partial_{\lambda}H_{\lambda}(\tau),\,\Delta\rho(\tau)]=0,\,\forall\tau\in[0,\,t_{f}],\label{eq:dHdlamb-delrho-commute}
\end{equation}
where we use the Liouville equation for $\Delta\rho(\tau)$. Since $\ket{\varphi_{a,\,b}(\tau)}\equiv U(\tau)\ket{\varphi_{a,\,b}}$ are associated with the non-degenerate
eigenvalues $\pm1$ of $\Delta\rho(\tau)$, we conclude that $\ket{\varphi_{a,\,b}(\tau)}$
must also be eigenvectors of $\partial_{\lambda}H_{\lambda}(\tau)$
at all times. That is, $\partial_{\lambda}H_{\lambda}(\tau)\ket{\varphi_{a,\,b}(\tau)}=\nu_{a,\,b}(\tau)\ket{\varphi_{a,\,b}(\tau)},\,\forall\tau\in[0,\,t_{f}],$
where $\nu_{a,\,b}(\tau)$ is the eigenvalue of $\partial_{\lambda}H_{\lambda}(\tau)$.
It is then straightforward to check that $\ket{\varphi_{a,\,b}}$  are eigenvectors
of $G_{t_{f}}[U]$ with eigenvalue $\int_{0}^{t_{f}}\nu_{a,\,b}(\tau)d\tau$.
Thus any unitary dynamics that preserves the adiabatic
evolution of any pair of eigenstates of $\partial_{\lambda}H_{\lambda}(\tau)$
satisfies Eq.~\eqref{eq:dHdlamb-delrho-commute} and is an extremal solution satisfying $\delta S_{\text{M}}=0$.
To further maximize the QFI among the manifold of extremal solutions, one needs
to further optimize the difference between $\int_{0}^{t_{f}}\nu_{a}(\tau)$
and $\int_{0}^{t_{f}}\nu_{b}(\tau)$. This  requires $\nu_{a}(\tau)$
and $\nu_{b}(\tau)$ to be the maximal and minimum eigenvalues of
$\partial_{\lambda}H_{\lambda}(\tau)$ at all times. When the unitary
dynamics preserves the adiabatic evolution of all eigenstates of $\partial_{\lambda}H_{\lambda}(\tau)$,
i.e., $U(\tau)=\sum_{\alpha}\ket{\varphi_{\alpha}(\tau)}\bra{\varphi_{\alpha}(0)}$,
where $\ket{\varphi_{\alpha}}$ denotes the eigenvectors of $\partial_{\lambda}H_{\lambda}(\tau)$,
one recovers   the Pang-Jordan control Hamiltonian~\cite{pang2017optimal}
$
H_{\text{c}}(\tau)=\text{i}\sum_{\alpha}\ket{\dot{\varphi}_{\alpha}(\tau)}\bra{\varphi_{\alpha}(\tau)}-H_{\lambda}(\tau).
$
Note that if there is any level crossing in $\nu_{a}(\tau)$ at some
instant time in $[0,\,t_{f}]$,  the way of labeling the eigenstates
and eigenvalues is not unique. Choosing $\nu_{+}(\tau)$ and $\nu_{-}(s)$  always as the maximum and
minimum eigenvalues of $\partial_{\lambda}H_{\lambda}(\tau)$ at all
times, the resulting QFI is the greatest among all the different
ways of labeling the eigenstates. With this labeling,
the first-order time derivative of $\ket{\varphi_{a}(\tau)}$ is discontinuous,
which results in a $\delta$-pulse in the control Hamiltonian $H_{\text{c}}$.
This provides an alternative understanding of the $\sigma_{x}$-like
pulses in Ref.~\cite{yang2017quantum,pang2017optimal,naghiloo2017achieving}.

In the general case, $\{\mathcal{X}_{i}\}_{i=1}^{d_{\text{c}}}$
do not span the whole space of Hermitian operators, and the generic optimal solutions $U(\tau)$ and $\ket{\varphi_{a,\,b}}$
implicitly depend on the probe time $t_{f}$. This is due to the dependence of the
generator $G_{t_{f}}[U]$  on $t_{f}$, that may make the eigenvectors
$\ket{\varphi_{a,\,b}}$ determining the optimal initial state also dependent on $t_{f}$. The dependence on $t_{f}$ for the optimal
unitary $U(\tau)$  can then be seen from Eq.~(\ref{eq:no-leakage}).
With these observations, we take the derivative with respective to $t_{f}$ to obtain the following differential-integral equation
\begin{equation}
\partial_{t_{f}}G_{t_{f}}\ket{\varphi_{\alpha,\,t_{f}}}+G_{t_{f}}\ket{\partial_{t_{f}}\varphi_{\alpha,\,t_{f}}}=\partial_{t_{f}}\mu_{\alpha,\,t_{f}}\ket{\varphi_{\alpha,\,t_{f}}}+\mu_{\alpha,\,t_{f}}\ket{\partial_{t_{f}}\varphi_{\alpha,\,t_{f}}},\label{eq:diff-eigen-condition}
\end{equation}
where we have suppressed the dependence on $U$ in the generator $G_{t_{f}}$ for simplicity, and $\alpha=a,\,b$,
$\partial_{t_{f}}G_{t_{f}}=\allowbreak\text{i}\underline{\partial}_{t_{f}}U_{t_{f}}^{\dagger}(t_{f})\partial_{\lambda}U_{t_{f}}(t_{f})$$+ \text{i}U_{t_{f}}^{\dagger}(t_{f})\underline{\partial}_{t_{f}}\partial_{\lambda}U_{t_{f}}(t_{f})$$+ U_{t_{f}}^{\dagger}(t_{f})\partial_{\lambda}H_{\lambda}(t_{f}) U_{t_{f}}(t_{f})$,
where $\underline{\partial}_{t_{f}}$ denotes the derivative with respect
the subscript $t_{f}$ instead of the one in the parenthesis~\cite{yang2021SM}. Generally,
Eq.~(\ref{eq:diff-eigen-condition}) is difficult to solve analytically,
while numerical calculation is  tractable. However, if $U_{t_{f}}(\tau)$
and $\ket{\varphi_{a,\,t_{f}}}$ are independent of the subscript variable $t_{f}$, Eq.~(\ref{eq:diff-eigen-condition})
reduces to $U^{\dagger}(t_{f})\partial_{\lambda}H_{\lambda}(t_{f})U(t_{f})\ket{\varphi_{\alpha}}=\partial_{t_{f}}\mu_{\alpha,\,t_{f}}\ket{\varphi_{\alpha}}$.
It then follows that $\ket{\varphi_{a}(t_{f})}=U(t_{f})\ket{\varphi_{a}}$
is an eigenstates of $\partial_{\lambda}H_{\lambda}(t_{f})$ for all
$t_{f}$, with eigenvalue $\partial_{t_{f}}\mu_{\alpha,\,t_{f}}$. The solution reduces again to the Pang-Jordan protocol~\cite{pang2017optimal}. 

This suggests that when the control Hamiltonian is restricted to some
non-trivial subspace of the Hermitian operators, such that $\partial_{\lambda}\Delta\rho(\tau)$
does not always vanish on $[0,\,t_{f}]$, both the \textit{exact}
optimal controls and the \textit{exact} optimal initial states
depend on the probing time $t_{f}$, making it challenging to 
find them analytically. In particular, one can show that when only $U(\tau)$ depends on
$t_{f}$ but $\ket{\varphi_{\alpha}}$ is independent of $t_{f}$, this is the case as
there exists a time such that $\partial_{\lambda}\Delta\rho(\tau)\neq0$~\cite{yang2021SM} . 
\begin{figure}
\begin{centering}
\includegraphics[width=\linewidth]{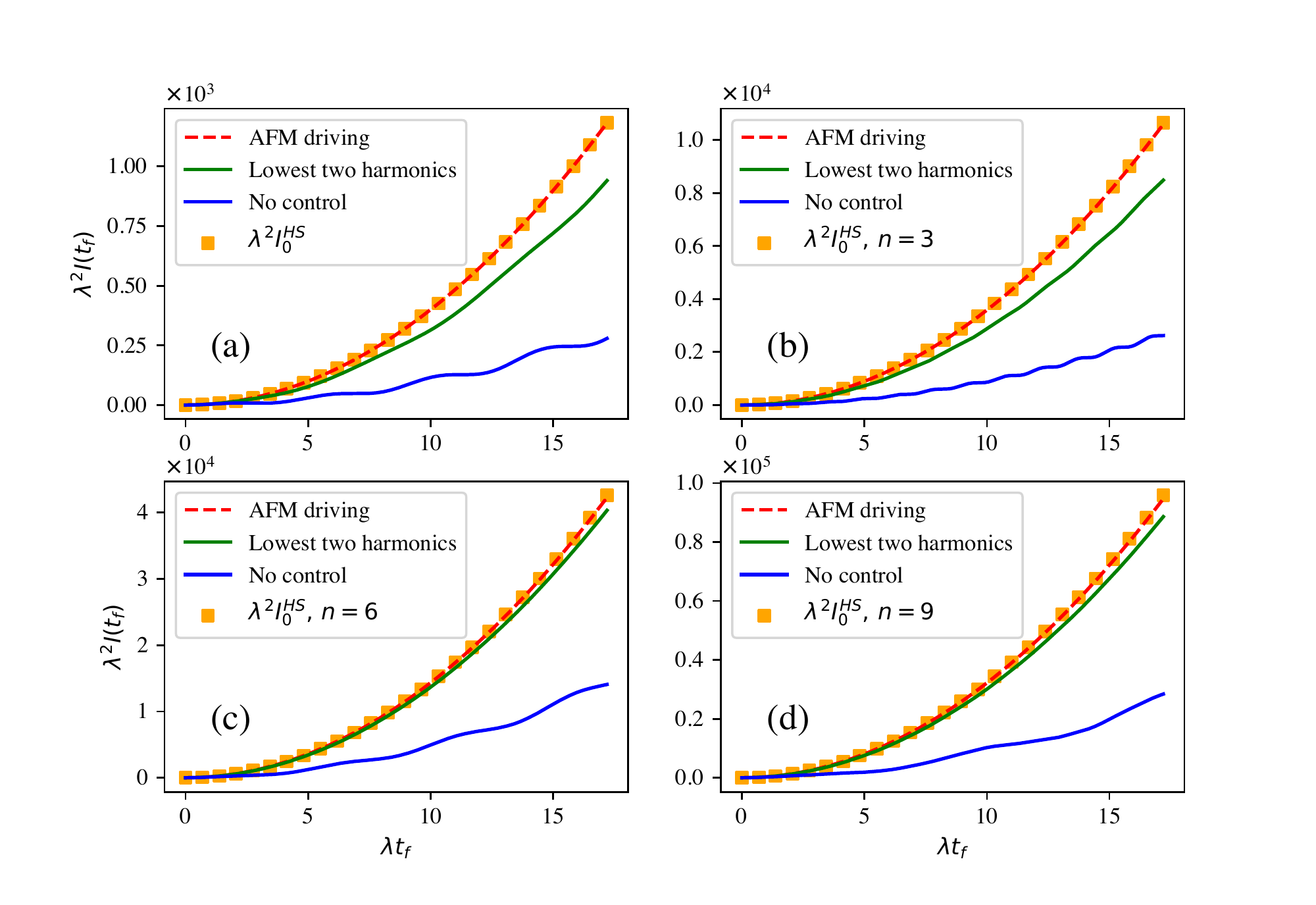}
\par\end{centering}
\caption{\label{fig:I-vs-t}QFI for various scenarios. (a) a single qubit with the sensing Hamiltonian~$H_{\lambda}=\lambda\sigma_{z}/2+\Delta\sigma_{x}/2$
(b-d) the spin chain 
with the
sensing Hamiltonian~(\ref{eq:spin-chain}). For all figures (a-d)
$\lambda=\Delta=1$, the frequency of the drive $\omega=1826.67$.
The red lines satisfy the AFM condition~\eqref{eq:FMC-qubit} or~\eqref{eq:FMC-spin-chain} with $c_{l}^{y}=\tilde{c}_{l}^{z}=c_{l}^{xy}=\tilde{c}_{l}^{zx}=10$
when $1\le l\le5$ and vanish for $l>5$. The total simulation time
is $5000$ times the fundamental periodic $2\pi/\omega$. For the
case of lowest two harmonics in green lines, $c_{l}^{y},\,\tilde{c}_{l}^{z},\,c_{l}^{xy},\,\tilde{c}_{l}^{zx}$
are non-vanishing only when $l=1,\,2$ and take values $10$. }
\end{figure}

\textit{Approximate solution via Floquet
engineering. } We next show that for a time-independent
Hamiltonian $H_{\lambda}$ the restricted optimal
controls can be \textit{approximately} engineered by the high-frequency
control, known as the Floquet engineering~\cite{rahav2003timeindependent,goldman2014periodically,bukov2015universal,eckardt2015highfrequency,goldman2015erratum,eckardt2017colloquium,chen2020textbackslash}.
For time-independent Hamiltonians with unrestricted controls, 
the minimum optimal control Hamiltonian is the one that
cancels the parts in $H_{\lambda}$ that do not commute with $\partial_{\lambda}H_{\lambda}$~\cite{pang2017optimal}.
However, if any of the required control Hamiltonians is not available, 
it is not clear how to reach the Heisenberg scaling, $I_0^{\textrm{HS}}=4n^2t_f^2$,  where $n$ is the number of probes. Here,
we search for a static control $H_{\text{c},\,0}$ and high-frequency driving controls $H_{\text{c}}^{(d)}(t)$, where $H_{\text{c}}^{(d)}(t)\equiv\sum_{l\neq0}H_{\text{c},\,l}e^{\text{i}l\omega t}$, 
so that $H_{\text{tot}}(t)=H_{\lambda}+H_{\text{c},\,0}+H_{\text{c}}^{(d)}(t)$, and show that
the unavailable controls in the minimum optimal control Hamiltonian can be actually constructed approximately
through a high-frequency expansion. Let us first move to the rotating
frame associated with $\mathcal{U}(t)=e^{\text{i}K(t)}$, where $K(t)$
is the so-called kick operator~\cite{goldman2014periodically}.
In it, the total Hamiltonian becomes is given by the Floquet effective
Hamiltonian $H_{\text{F}}$. When the driving frequency is high
enough, an expansion in orders of $1/\omega$ can be performed~\cite{goldman2014periodically,goldman2015erratum}.
To the first-order of $1/\omega$, one  finds that $K(t)=[1/(\text{i}\omega)]\sum_{l\neq0}\frac{1}{l}\left(H_{\text{c},\,l}e^{\text{i}l\omega t}-1\right)+O(1/\omega^{2})$~\footnote{We choose the normalization $K(0)=0$ to all orders of $1/\omega$,
which is different from the normalization $1/T\int_{0}^{T}K(t)dt=0$
used in Ref.~\cite{goldman2014periodically,goldman2015erratum}.
Therefore, the resulting expression of $K(t)$ is different from the
one in Ref.~\cite{goldman2014periodically,goldman2015erratum}, up
to some irrelevant constant, which does not affect the form of the
Floquet effective Hamiltonian $H_{\text{F}}$} and
\begin{equation}
H_{\text{F}}=H_{\lambda}+H_{\text{c},\,0}+1/\omega\sum_{l=1}^{\infty}[H_{\text{c},\,l},\,H_{\text{c},\,-l}]/l+O(1/\omega^{2}).\label{eq:HF}
\end{equation}
One can explicitly show that the kick operator $K(t)$
consists of the $H_{\text{c},\,l}$ with $l\neq0$ and  it
is independent of the estimation parameter. Thus, the QFI remains unchanged in the Floquet rotating frame.  The generator becomes $G_{t_{f}}=\int_{0}^{t_{f}}e^{-\text{i}H_{\text{F}}\tau}\partial_{\lambda}H_{\lambda}e^{\text{i}H_{\text{F}}\tau}d\tau$~\cite{yang2021SM}.
The key idea is that the unavailable controls in the original
static frame can be constructed through the commutator in Eq.~(\ref{eq:HF}).
This can be best illustrated using a simple qubit example with
the Hamiltonian $H_{\lambda}=\lambda\sigma_{z}/2+\Delta\sigma_{x}/2$ and $\mathcal{V}_{\text{c}}=\{\sigma_{y},\,\sigma_{z}\}$. The term
that does not commute with $\partial_{\lambda}H_{\lambda}$ is $\Delta\sigma_{x}/2$,
which is  not available in $\mathcal{V}_{\text{c}}$. Therefore, we consider
$H_{\text{c},\,0}=c_{0}^{y}\sigma_{y}+c_{0}^{z}\sigma_{z}$ and $H_{\text{c}}^{(d)}(t)=\sum_{l\neq0}(c_{l}^{y}\sigma_{y}+c_{l}^{z}\sigma_{z})e^{\text{i}l\omega t}$,
where for $l\neq0$, $c_{l}^{y}=c_{-l}^{y*}$ and $c_{l}^{z}=c_{-l}^{z*}$
to guarantee the Hermiticity of $H_{\text{c}}^{(d)}(t)$. According
to Eq.~(\ref{eq:HF}),  to
the first order of $1/\omega$, we find 
\begin{equation}
H_{F}=(\lambda/2+c_{0}^{z})\sigma_{z}+(\Delta/2-4/\omega\sum_{l=1}^{\infty}\text{Im}[c_{l}^{y}c_{l}^{z*}]/l)\sigma_{x}+c_{0}^{y}\sigma_{y}.
\end{equation}
The approximate optimal control fulfills $c_{0}^{y}=0$ and 
\begin{equation}
\omega=8/\Delta\sum_{l=1}^{\infty}\text{Im}[c_{l}^{y}c_{l}^{z*}]/l,\label{eq:FMC-qubit}
\end{equation}
which we call \textit{the amplitude-frequency matching} (AFM) condition. The validity of the high-frequency expansion requires that $\omega$
should be the largest frequency in the original Hamiltonian, i.e., that $\omega\gg\lambda,\,\Delta,c_{l}^{y},c_{l}^{z}$.
Experimentally, for a laser frequency that satisfies this condition, one can always tune the amplitude so that Eq.~\eqref{eq:FMC-qubit} is satisfied. Conversely, one can also choose a proper laser frequency for given amplitudes so that Eq.~\eqref{eq:FMC-qubit} is fulfilled. We emphasize that when Eq.~\eqref{eq:FMC-qubit} is satisfied and the initial state in the Floquet rotating frame is $(\ket{0}+\ket{1})/\sqrt{2}$, the optimal control conditions Eqs.~(\ref{eq:G-eigen}, \ref{eq:no-leakage}) are \textit{approximately} satisfied up to the order of $1/\omega$ in the Floquet rotating frame.  The initial state in the lab and Floquet rotating frames is the same as  $K(0)=0$. Going back to  the lab frame, we generate a solution to Eqs.~(\ref{eq:G-eigen}, \ref{eq:no-leakage}) when the controls are restricted. This illustrates the power of our approach  beyond the Pang-Jordan protocol~\cite{pang2017optimal}. 

A simple choice for the amplitudes involves taking for $l\ge1$ both $c_{l}^{y}$ and $\tilde{c}_{l}^{z}=\text{i}c_{l}^{z}$ to be real. This yields \begin{equation}
H_{\text{c}}^{(d)}(t)=2\sum_{l=1}^{\infty}\left[c_{l}^{y}\cos(l\omega t)\sigma_{y}+\tilde{c}_{l}^{z}\sin(l\omega t)\sigma_{z}\right].
\end{equation}
As one can see from Fig.~\ref{fig:I-vs-t}(a),
for parameters satisfying the AFM condition Eq.~\eqref{eq:FMC-qubit}, the Heisenberg
scaling is achieved. 
As a result, the first-order term in the high-frequency expansion makes it possible to construct the $\sigma_{x}$ term by
 commuting the operators in $\mathcal{V}_{\text{c}}$. This idea
can be generalized to many-body systems, as shown in the following.

\textit{Restricted control in a quantum spin chain.} Consider the sensing of magnetic field using a spin chain
\begin{equation}
H_{\lambda}=\frac{J}{2}\sum_{i=1}^{n}\sigma_{i}^{x}\sigma_{i+1}^{x}+\frac{\Delta}{2}\sum_{i=1}^{n}\sigma_{i}^{x}\sigma_{i+1}^{x}\sigma_{i+2}^{x}+\frac{\lambda}{2}\sum_{i=1}^{n}\sigma_{i}^{z},\label{eq:spin-chain}
\end{equation}
which contains both two-body and three-body interactions. We assume
 periodic boundary conditions for simplicity and consider the set of allowed controls $\mathcal{V}_{\text{c}}=\{\sigma_{i}^{a},\,\sigma_{i}^{a}\sigma_{i+1}^{b}\}$,
$(a,\,b=x,\,y,\,z)$, involving only one-body and nearest
neighbor two-body operators. When the controls are unrestricted,
the minimum optimal control Hamiltonian  contains local two-body and local three-body terms.  The first-part in the control Hamiltonian
would cancel the nearest neighbor two-body terms,
i.e., $H_{\text{c},\,0}=-J/2\sum_{i=1}^{n}\sigma_{i}^{x}\sigma_{i+1}^{x}$.
However, the term consisting of three-body operators cannot be canceled
directly through the allowed control set $\mathcal{V}_{\text{c}}$. 
\begin{figure}
\begin{centering}
\includegraphics[width=0.8\linewidth]{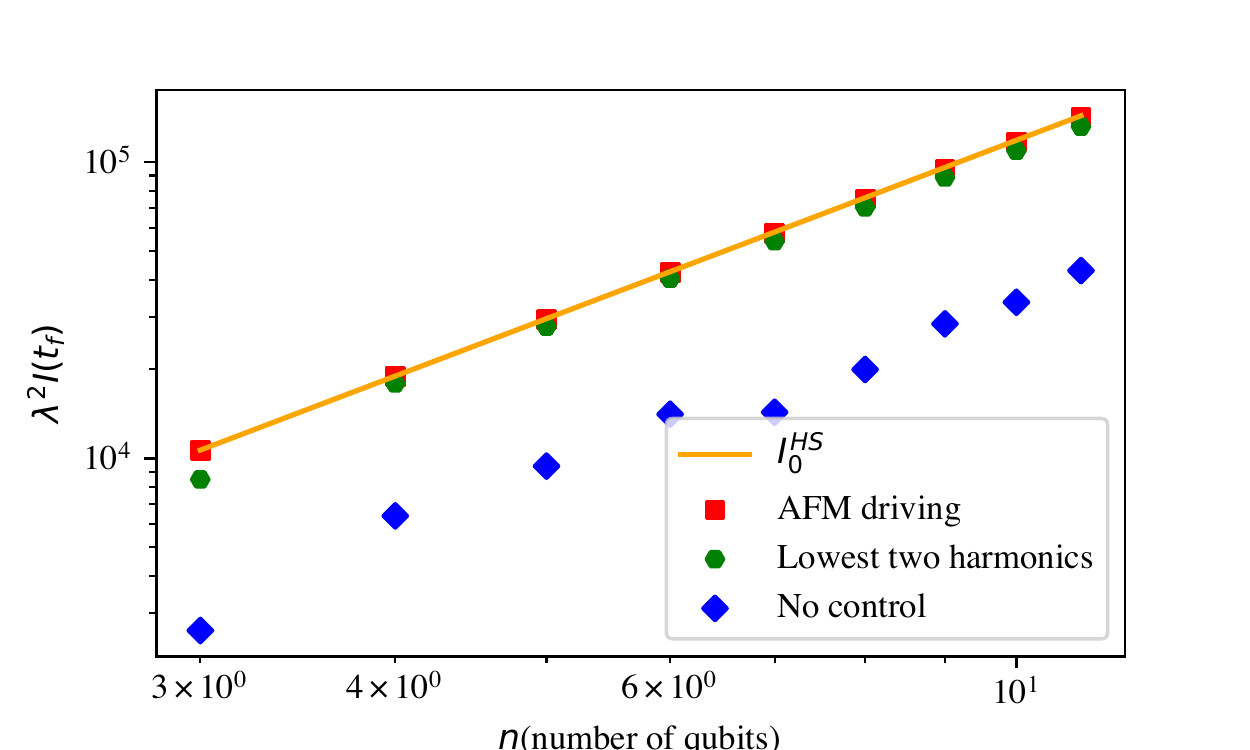}
\par\end{centering}
\caption{\label{fig:HS}Numerical calculation of the QFI with respect to the number of qubit $n$ for the sensing
of magnetic field using the spin-chain Hamiltonian~(\ref{eq:spin-chain}). 
The initial state preparation and the values of the parameters $\lambda$,
$\Delta$, $\omega$, $c_{l}^{xy}$ and $\tilde{c}_{l}^{zx}$ are
the same as Fig.~\ref{fig:I-vs-t} with $J=0$. (b-d) and $t_{f}=17.2$ is the
total simulation time in Fig.~\ref{fig:I-vs-t}.}
\end{figure}

Our goal is to construct the three-body operators using the first-order
correction in Eq.~(\ref{eq:HF}), that is, we expect to induce $[H_{l},\,H_{-l}]\propto\sum_{i=1}^{n}\sigma_{i}^{x}\sigma_{i+1}^{x}\sigma_{i+2}^{x}$.
The commutator between
one-body and two-body operators cannot produce a three-body operator.
To generate a three-body operator, we must commute
the two-body operators in $\mathcal{V}_{\text{c}}$~\cite{yang2021SM}. Therefore, one can construct
$H_{\text{c}}(t)=\sum_{l\neq0}H_{\text{c},\,l}e^{\text{i}l\omega t}$,
where $H_{\text{c},\,l}=\sum_{i=1}^{n}\left(c_{li}^{xy}\sigma_{i}^{x}\sigma_{i+1}^{y}+c_{li,zx}\sigma_{i}^{z}\sigma_{i+1}^{x}\right)$,
$c_{li}^{xy}=c_{-li}^{xy*}$ and $c_{li}^{zx}=c_{-li}^{zx*}$
to ensure the Hermiticity of $H_{\text{c}}(t)$~\cite{yang2021SM}. When $c_{li}^{xy}$
and $c_{li}^{zx}$ are respectively real and purely imaginary, the
commutator $[H_{\text{c},\,l},\,H_{\text{c},\,-l}]=-4\sum_{i=1}^{n}\text{Im}(c_{li}^{xy}c_{li+1}^{zx*})\sigma_{i}^{x}\sigma_{i+1}^{x}\sigma_{i+2}^{x}$.
The most general choice of the coefficients $c_{li}^{xy}$ and $c_{li}^{zx}$
which generates the $\sigma^{x}$-type three-body interaction is provided in 
Sec.~V in ~\cite{yang2021SM}. For
simplicity, we further assume these coefficients are homogeneous across
the chain and take $c_{li}^{xy}=c_{l}^{xy}$ and $\tilde{c}_{l}^{zx}=\text{i}c_{li}^{zx}$ to be real
for $l\ge1$. This yields the high frequency driving control Hamiltonian 
\begin{align}
H_{\text{c}}^{(d)}(t) & =2\sum_{l=1}^{\infty}\sum_{i=1}^{n}\left[c_{l}^{xy}\cos(l\omega t)\sigma_{i}^{x}\sigma_{i+1}^{y}+\tilde{c}_{l}^{zx}\sin(l\omega t)\sigma_{i}^{z}\sigma_{i+1}^{x}\right],\label{eq:Hc-d-spin}
\end{align}
while the effective Hamiltonian obtained as the leading term in the high-frequency expansion becomes
\begin{equation}
H_{F}=\frac{\lambda}{2}\sum_{i=1}^{n}\sigma_{i}^{z}+\left[\frac{\Delta}{2}-\frac{4}{\omega}\sum_{l=1}^{\infty}\frac{1}{l}(c_{l}^{xy}\tilde{c}_{l}^{zx})\right]\sigma_{i}^{x}\sigma_{i+1}^{x}\sigma_{i+2}^{x}.
\end{equation}
The three-body term is then canceled by tuning $c_{l}^{xy}$ and $\tilde{c}_{l}^{zx}$ such that 
\begin{equation}
\omega=8/\Delta\sum_{l=1}^{\infty}(c_{l}^{xy}\tilde{c}_{l}^{zx}/l).\label{eq:FMC-spin-chain}
\end{equation}
As with the qubit case, the optimal initial state in both the lab and
Floquet rotating is the GHZ state. Equations~(\ref{eq:G-eigen}, \ref{eq:no-leakage}) are thus \textit{approximately} satisfied up to the order of $1/\omega$. As shown in Fig.~\ref{fig:I-vs-t}(b-d),
for parameters satisfying the AFM condition~(\ref{eq:FMC-spin-chain}), the
Heisenberg scaling $I_0^{\textrm{HS}}$ is achieved. Furthermore,
even when one just takes the lowest two harmonics in the driving,
QFI is not very far below $I_0^{\textrm{HS}}$. In Fig.~\ref{fig:HS},
$I_0^{\textrm{HS}}$ is achieved with Eq.~\eqref{eq:Hc-d-spin} when the parameters satisfy Eq.~(\ref{eq:FMC-spin-chain}). Again, the precision using only the lowest two harmonics can approach $I_0^{\textrm{HS}}$ very closely. 

Finally, we emphasize that a similar
technique can be applied to design more general driving protocols
to cancel the effect of other types of three-body interactions~\cite{yang2021SM}.  Experimental platforms where three-body interactions either appear naturally or can be potentially engineered -such as NMR system \cite{tseng1999quantum,peng2009quantum}, Kitaev spin liquid \cite{takahashi2021topological}, superconducting circuit \cite{petiziol2021quantum} and quantum gas systems \cite{de2013nonequilibrium,signoles2021glassy,vishveshwara2021z,buchler2007three}- can be potentially used to test the metrological protocol discussed here. 

In summary, we have introduced a variational approach to quantum parameter estimation and derived the optimal control
equations under which the precision is optimal when the available control Hamiltonians are limited. This approach readily yields the optimal initial state and Hamiltonian controls, that are generally dependent on the probe time, in contrast with the unconstrained case. The implementation of the constrained optimal protocol in many-body systems can be eased by Floquet engineering, as we have demonstrated in applications to magnetometry.
We hope that  our results 
inspire new theoretical and technological advances in quantum
metrology with quantum many-body systems. Many questions are open for further
investigation, such as determining the ultimate scaling bounds of the QFI under
restricted local controls and the application of our method to critically-enhanced quantum metrology \cite{Rams18}.

\textit{Acknowledgement}. We are grateful to Aur\'elia Chenu for critical reading of the manuscript and helpful feedback. We thank Nicolas P. Bigelow and Fernando J. G\'omez-Ruiz for useful
discussions. Support from National Natural Science Foundation of China
(NSFC) grant no.~12075323,  NSF grant PHY~1708008,
NASA/JPL RSA~1656126, and US Army Research Office grant no.~W911NF-18-10178 is greatly acknowledged.

\let\oldaddcontentsline\addcontentsline     
\renewcommand{\addcontentsline}[3]{}         

\bibliographystyle{apsrev4-1}
\bibliography{ManyBodyMetrology}

\clearpage\newpage\setcounter{equation}{0} \setcounter{section}{0}
\setcounter{subsection}{0} 
\global\long\def\theequation{S\arabic{equation}}%
\onecolumngrid \setcounter{enumiv}{0} 

\setcounter{equation}{0} \setcounter{section}{0} \setcounter{subsection}{0} \renewcommand{\theequation}{S\arabic{equation}} \onecolumngrid \setcounter{enumiv}{0}
\begin{center}
\textbf{\large{}Supplemental Materials}{\large\par}
\par\end{center}

\let\addcontentsline\oldaddcontentsline     

\tableofcontents{}

\addtocontents{toc}{\protect\thispagestyle{empty}}
\pagenumbering{gobble}

\section{\label{sec:Initial-State-Opt}Optimization over the initial state}

Our goal is to maximize $\text{Var}[\mathcal{O}]\big|_{\ket{\varphi_{0}}}$
over all the quantum states $\ket{\varphi_{0}}$, where $\mathcal{O}$
is some Hermitian operator. We note the following gauge invariance

\begin{align}
\mathcal{O} & \to\tilde{\mathcal{O}}\equiv\mathcal{O}-c,\\
\text{Var}[\mathcal{O}]\big|_{\ket{\varphi_{0}}} & \to\text{Var}[\tilde{\mathcal{O}}]\big|_{\ket{\varphi_{0}}}=\text{Var}[\mathcal{O}]\big|_{\ket{\varphi_{0}}},
\end{align}
where $c$ is a constant. Introducing a Lagrange multiplier to account
for the normalization of $\ket{\varphi_{0}}$, we aim at optimizing
the following function 
\begin{align}
f & \equiv\text{Var}[\tilde{\mathcal{O}}]\big|_{\ket{\varphi_{0}}}-\mu\braket{\varphi_{0}\big|\varphi_{0}}\nonumber \\
 & =\braket{\varphi_{0}\big|\tilde{\mathcal{O}}^{2}\big|\varphi_{0}}-\braket{\varphi_{0}\big|\tilde{\mathcal{O}}\big|\varphi_{0}}^2-\mu\braket{\varphi_{0}\big|\varphi_{0}}.
\end{align}
One can choose $c=\braket{\varphi_{0}\big|\mathcal{O}\big|\varphi_{0}}$,
so that $f$ becomes 
\begin{equation}
f=\braket{\varphi_{0}\big|\tilde{\mathcal{O}}^{2}\big|\varphi_{0}}-\mu\braket{\varphi_{0}\big|\varphi_{0}}.
\end{equation}
Taking derivatives with respect to $\ket{\varphi_{0}}$, one can easily
obtain 
\begin{equation}
\tilde{\mathcal{O}}^{2}\ket{\varphi_{0}}=\mu\ket{\varphi_{0}}.
\end{equation}
All eigenstates of $\mathcal{\tilde{\mathcal{O}}}$ must
be an eigenstate of $\tilde{\mathcal{O}}^{2}$, but not vice versa.
This is because if there is a pair of opposite eigenvalues of $\mathcal{\tilde{\mathcal{O}}}$,
the eigenstate of $\tilde{\mathcal{O}}^{2}$ could be the superposition
of the eigenstates of $\tilde{\mathcal{O}}$ corresponding to these
two opposite eigenvalues. Therefore, either $\ket{\varphi_{0}}$ is
 an eigenstate of $\tilde{\mathcal{O}}$ or $\ket{\varphi_{0}}$
is the superposition of two eigenstates of $\tilde{\mathcal{O}}$
with opposite eigenvalues. In the former case $f=0$ corresponds
to the extremal minimum value. In the latter case, i.e., 
\begin{equation}
\ket{\varphi_{0}}=\alpha\ket{\varphi_{a}}+\beta\ket{\varphi_{b}},
\end{equation}
where 
\begin{align}
\mathcal{O}\ket{\varphi_{a}} & =\mu_{a}\ket{\varphi_{a}},\\
\mathcal{O}\ket{\varphi_{b}} & =\mu_{b}\ket{\varphi_{b}},
\end{align}
$\mu_{a},\,\mu_{b}$ are not equal to each other and satisfy
\begin{equation}
\mu_{a}+\mu_{b}-2\braket{\varphi_{0}\big|\mathcal{O}\big|\varphi_{0}}=0.
\end{equation}
Then, the function $f$ becomes 
\begin{equation}
f(a,\,b)=a\mu_{a}^{2}+b\mu_{b}^{2}-(a\mu_{a}+b\mu_{b})^{2},
\end{equation}
where $a\equiv|\alpha|^{2}$ and $b\equiv|\beta|^{2}$.  We proceed
to introduce a Lagrangian multiplier to perform further maximization
over $a$ and $b$ 
\begin{equation}
g(a,\,b)=a\mu_{a}^{2}+b\mu_{b}^{2}-(a\mu_{a}+b\mu_{b})^{2}+\xi(a+b-1).
\end{equation}
Setting $\partial_{a}g(a,\,b)=\partial_{b}g(a,\,b)=0$, one find $a=b=1/2$.
Therefore, we find $f(a,\,b)=(\mu_{a}-\mu_{b})^{2}/4$. Finally, optimizing
over all the possible distinct eigenvalues of $\mathcal{O}$, we find
the global maximum values of $f$. That is, $f_{\max}=(\mu_{+}-\mu_{-})^{2}/4$,
with the corresponding optimal initial state  $\ket{\varphi_{0}}=(\ket{\varphi_{+}}+e^{\text{i}\theta}\ket{\varphi_{-}})/\sqrt{2}$,
where $\ket{\varphi_{\pm}}$ correspond to the maximum and minimum
eigenvalues $\mu_{\pm}$, respectively, and $\theta$ is an arbitrary
phase.

\section{\label{sec:Variational-Calculus}Variational Calculus}

For the variation with respect to $U$ we have 
\begin{equation}
\delta S_{\textrm{M}}=\delta S_{\textrm{I}}+\delta S_{\textrm{S}}=0,\label{eq:delSM-delU}
\end{equation}
where 
\begin{align}
\delta S_{\textrm{I}} & =-\int_{0}^{t_{f}}\text{Tr}\left\{ \Delta\rho U^{\dagger}(\tau)\delta U(\tau)U^{\dagger}(\tau)\partial_{\lambda}H_{\lambda}(\tau)U(\tau)\right\} d\tau+\int_{0}^{t_{f}}\text{Tr}\left\{ \Delta\rho U^{\dagger}(\tau)\partial_{\lambda}H_{\lambda}(\tau)\delta U(\tau)\right\} d\tau\nonumber \\
 & =\int_{0}^{t_{f}}\text{Tr}\left\{ U^{\dagger}[\Delta\rho(\tau),\,\partial_{\lambda}H_{\lambda}(\tau)]\delta U(\tau)\right\} d\tau,\label{eq:del-SI}\\
\delta S_{\textrm{S}} & =\text{i}\text{Tr}\left\{ U^{\dagger}(\tau)\Lambda(\tau)\delta U(\tau)\right\} \big|_{\tau=0}^{\tau=t_{f}}-\text{i}\int_{0}^{t_{f}}d\tau\text{Tr}\left\{ \left[U^{\dagger}(\tau)\Lambda(\tau)\dot{U}(\tau)U^{\dagger}(\tau)+\frac{d}{d\tau}\left(U^{\dagger}(\tau)\Lambda(\tau)\right)\right]\delta U(\tau)\right\} \nonumber \\
 & =\text{i}\text{Tr}\left\{ U^{\dagger}(\tau)\Lambda(\tau)\delta U(\tau)\right\} \big|_{\tau=0}^{\tau=t_{f}}-\text{i}\int dt\text{Tr}\left\{ U^{\dagger}\left[\Lambda\dot{U}U^{\dagger}+\dot{\Lambda}+U\dot{U}^{\dagger}\Lambda\right]\delta U(\tau)\right\} ,
\end{align}
where we have used the fact that $\delta U^{\dagger}=-U^{\dagger}\delta UU^{\dagger}$
in Eq.~(\ref{eq:del-SI}). Eq.~(\ref{eq:delSM-delU}) yields 
\begin{equation}
\dot{\Lambda}(\tau)-\text{i}[\Lambda(\tau),\,H_{\text{tot}}(\tau)]+\text{i}[\Delta\rho(\tau),\,\partial_{\lambda}H_{\lambda}(\tau)]=0,\label{eq:EL-Uc}
\end{equation}
where 
\begin{align}
\Delta\rho(\tau) & \equiv U(\tau)\Delta\rho U^{\dagger}(\tau),\\
H_{\text{tot}}(\tau) & \equiv\text{i}\dot{U}(\tau)U^{\dagger}(\tau)=H_{\lambda}(\tau)+H_{\text{c}}(\tau),
\end{align}
with the essential boundary condition 
\begin{equation}
\Lambda(t_{f})=0.\label{eq:Lambda-BC}
\end{equation}

We observe that Eq.~(\ref{eq:EL-Uc}) may be solved in terms of the
evolution operator because Eq.~(\ref{eq:EL-Uc}) is a driven Liouville
equation and the Green's function is the product of unitary evolution
operators for the forward and backward directions multiplied by the
Heaviside function. To see this in a clear manner, let us note  that
\begin{equation}
\Lambda_{\text{hom}}(\tau)=U(\tau,\,t_{f})\Lambda(t_{f})U^{\dagger}(\tau,\,t_{f})=0,
\end{equation}
where $U(\tau,\,s)$ is the evolution operator from time $s$ to time
$\tau$, which satisfies 
\begin{equation}
\text{i}\dot{U}(\tau,\,s)=H_{\text{tot}}(\tau)U(\tau,\,s),
\end{equation}
where the time derivative is applied to the variable $\tau$. One
can explicitly check that for the impulse applied at time $s$, the
Green's function is 
\begin{equation}
\mathscr{G}(\tau,\,s)=-U(\tau,\,s)\Theta(s-\tau)U^{\dagger}(\tau,\,s),\,s\in[0,\,t_{f}],\label{eq:Green}
\end{equation}
which satisfies 
\begin{equation}
\partial_{\tau}\mathscr{G}(\tau,\,s)=\text{i}[\mathscr{G}(\tau,\,s),\,H_{\text{tot}}(\tau)]+\text{\ensuremath{\delta}}(\tau-s),
\end{equation}
and the homogeneous initial condition $\mathscr{G}(t_{f},\,s)=0$.
Note that we do not contract $U(\tau,\,s)$ and $U^{\dagger}(\tau,\,s)$
in Eq.~(\ref{eq:Green}). Therefore, we find 
\begin{align}
\Lambda(\tau) & =\Lambda_{\text{inhom}}(\tau)=\text{i}\int_{0}^{t_{f}}U(\tau,\,s)[\Delta\rho(s),\,\partial_{\lambda}H_{\lambda}(s)]\Theta(s-\tau)U^{\dagger}(\tau,\,s)ds\nonumber \\
 & =\text{i}U(\tau)[\Delta\rho,\,\int_{\tau}^{t_{f}}U^{\dagger}(s)\partial_{\lambda}H_{\lambda}(s)U(s)ds]U^{\dagger}(\tau)\nonumber \\
 & =\text{i}U(\tau)[\Delta\rho,\,G_{t_{f}}[U]-G_{\tau}[U]]U^{\dagger}(\tau)\nonumber \\
 & =-\text{i}U(\tau)[\Delta\rho,\,G_{\tau}[U]]U^{\dagger}(\tau),\label{eq:Lambda-solution}
\end{align}
where we have used the fact that $[\Delta\rho,\,G_{t_{f}}[U]]=0$.
We note the alternative concise expression for the generator~\cite{pang2017optimal}
\begin{equation}
G_{\tau}[U]=\text{i}U^{\dagger}(\tau)\partial_{\lambda}U(\tau).\label{eq:G-alternative}
\end{equation}
Substituting Eq.~(\ref{eq:G-alternative}) into Eq.~(\ref{eq:Lambda-solution})
yields, 
\begin{align}
\Lambda(\tau) & =U(\tau)\Delta\rho U^{\dagger}(\tau)\partial_{\lambda}U(\tau)U^{\dagger}(\tau)-\partial_{\lambda}U(\tau)\Delta\rho U^{\dagger}(\tau)\nonumber \\
 & =-[U(\tau)\Delta\rho\partial_{\lambda}U^{\dagger}(\tau)+\partial_{\lambda}U(\tau)\Delta\rho U^{\dagger}(\tau)]=-\partial_{\lambda}[\Delta\rho(\tau)],\label{eq:Lambda-drhodlam}
\end{align}
where we have used $\partial_{\lambda}U(\tau)U^{\dagger}(\tau)=-U(\tau)\partial_{\lambda}U^{\dagger}(\tau)$. Therefore, the trace condition $\text{Tr}\left\{ \Lambda(\tau)\mathcal{X}_{i}\right\} =0$
becomes Eq.~(\ref{eq:no-leakage}) in the main text.

Mathematically, we note that $\Delta\rho$ is independent of the estimation of
$\lambda$ from the very beginning when we construct the metrological
action Eq.~(\ref{eq:metrology-action}) in the main text. Practically, after solving the optimization problem, one may find that the optimal value of $\Delta\rho$ depends on $\lambda$, since $\Delta\rho$ is related the maximum
and minimum eigenvectors of $G_{t_{f}}[U]$, which depends on $\lambda$ in general.  If this is the case, the optimal initial state should be prepared with the prior knowledge of the estimation
parameter, which is very close to the
true value of $\lambda$, as we focus on the ultra-sensitive estimation regime.

\section{\label{sec:tf-dependence} Dependence of $U_{t_{f}}(\tau)$ on the probe time $t_{f}$ }

Let us first calculate explicitly $\partial_{t_{f}}G_{t_{f}}$. We note the generator $G_{t_{f}}$ has the following two alternative forms
\begin{equation}
G_{t_{f}}=\text{i}U_{t_{f}}^{\dagger}(t_{f})\partial_{\lambda}U_{t_{f}}(t_{f})=\int_{0}^{t_{f}}U_{t_{f}}^{\dagger}(\tau)\partial_{\lambda}H_{\lambda}(\tau)U_{t_{f}}(\tau)d\tau.
\end{equation}
Then it is straightforward to calculate
\begin{align}
\partial_{t_{f}}G_{t_{f}} & =\int_{0}^{t_{f}}\underline{\partial}_{t_{f}}\left[U_{t_{f}}^{\dagger}(\tau)\partial_{\lambda}H_{\lambda}(\tau)U_{t_{f}}(\tau)\right]d\tau+U_{t_{f}}^{\dagger}(t_{f})\partial_{\lambda}H_{\lambda}(t_{f})U_{t_{f}}(t_{f})\nonumber \\
 & =\underline{\partial}_{t_{f}}\left[\text{i}U_{t_{f}}^{\dagger}(t_{f})\partial_{\lambda}U_{t_{f}}(t_{f})\right]+U_{t_{f}}^{\dagger}(t_{f})\partial_{\lambda}H_{\lambda}(t_{f})U_{t_{f}}(t_{f})\nonumber \\
 & =\text{i}\underline{\partial}_{t_{f}}U_{t_{f}}^{\dagger}(t_{f})\partial_{\lambda}U_{t_{f}}(t_{f})+\text{i}U_{t_{f}}^{\dagger}(t_{f})\underline{\partial}_{t_{f}}\partial_{\lambda}U_{t_{f}}(t_{f})+U_{t_{f}}^{\dagger}(t_{f})\partial_{\lambda}H_{\lambda}(t_{f})U_{t_{f}}(t_{f})
\end{align}
where $\underline{\partial}_{t_{f}}$ denotes the derivative with
respect the subscript $t_{f}$ instead of the one in the parenthesis.

When $\ket{\varphi_{a,\,t_{f}}}$ is independent of $t_{f}$, after
multiplying both sides by $U_{t_{f}}(t_{f})$ from the left, Eq.~(\ref{eq:diff-eigen-condition}) in the main text
becomes 
\begin{equation}
\partial_{t_{f}}\tilde{G}_{t_{f}}\ket{\varphi_{a,\,t_{f}}(t_{f})}=\partial_{t_{f}}\mu_{a,\,t_{f}}\ket{\varphi_{a,\,t_{f}}(t_{f})},\,\forall t_{f},\,a=\pm,\label{eq:eigenstates-partialGtilde}
\end{equation}
where $\ket{\varphi_{a,\,t_{f}}(t_{f})}=U_{t_{f}}(t_{f})\ket{\varphi_{a}}$,
and 
\begin{align}
\partial_{t_{f}}\tilde{G}_{t_{f}} & \equiv U_{t_{f}}(t_{f})\partial_{t_{f}}G_{t_{f}}U_{t_{f}}^{\dagger}(t_{f})\nonumber \\
 & =\text{i}\underline{\partial}_{t_{f}}U_{t_{f}}(t_{f})\partial_{\lambda}U_{t_{f}}^{\dagger}(t_{f})+\text{i}\underline{\partial}_{t_{f}}\partial_{\lambda}U_{t_{f}}(t_{f})U_{t_{f}}^{\dagger}(t_{f})+\partial_{\lambda}H_{\lambda}(t_{f})\nonumber \\
 & =\text{i}\partial_{\lambda}\left[\underline{\partial}_{t_{f}}U_{t_{f}}(t_{f})U_{t_{f}}^{\dagger}(t_{f})\right]+\partial_{\lambda}H_{\lambda}(t_{f})\nonumber \\
 & =\partial_{\lambda}\left[H_{\lambda}(t_{f})+A_{\lambda,\,t_{f}}(t_{f})\right].
\end{align}
Here, $A_{\lambda,\,t_{f}}(t_{f})$ is defined as 
\begin{equation}
A_{\lambda,\,t_{f}}(\tau)\equiv-\text{i}U_{t_{f}}(\tau)\underline{\partial}_{t_{f}}U_{t_{f}}^{\dagger}(\tau),
\end{equation}
which bears the same form as $G_{t_{f}}$ and therefore has the
integral representation 
\begin{equation}
A_{\lambda,\,t_{f}}(\tau)=-\int_{0}^{\tau}U_{t_{f}}^{\dagger}(s)\partial_{t_{f}}H_{\text{c},\,t_{f}}(s)U_{t_{f}}(s)ds.
\end{equation}
The physical meaning of $A_{\lambda,\,t_{f}}(\tau)$ is the following:
When $A_{\lambda,\,t_{f}}(\tau)$ vanishes as in the Pang-Jordan protocol~\cite{pang2017optimal},
Eq.~(\ref{eq:eigenstates-partialGtilde}) implies that $\ket{\varphi_{\alpha,\,t_{f}}(\tau)}$
is always an eigenstates of $\partial_{\lambda}H_{\lambda}(\tau)$,
as can be seen  from Eq.~(\ref{eq:dHdlamb-delrho-commute}) in the main text, given that $\Delta\rho_{t_{f}}(\tau)$
always commutes with $\partial_{\lambda}H_{\lambda}(\tau)$ for any
$\tau\in[0,\,t_{f}]$.  Eq.~(\ref{eq:eigenstates-partialGtilde})
indicates that, as long as long as $[\partial_\lambda A_{\lambda,\,t_{f}}(t_{f}),\,\partial_{\lambda}H(t_{f})]\neq0$,
$\Delta\rho_{t_{f}}(t_{f})$ does not commute with $\partial_{\lambda}H_{\lambda}(t_{f})$
in general, which in turn implies that $\partial_{\lambda}[\Delta\rho_{t_{f}}(\tau)]$
does not vanishes at all times. Therefore Eq.~(\ref{eq:no-leakage}) in the main text
becomes non-trivial in the sense that $\mathcal{X}_{j}$ must be orthogonal
to the non-vanishing values of $\partial_{\lambda}\Delta\rho_{t_{f}}(\tau)$
at least for some time $\tau$. 

\section{\label{sec:rotating-frame}Generator and optimal control in parameter-independent
rotating frame}

For a \textit{parameter-independent} unitary operator $\mathcal{U}(t)$,
consider the interaction frame associated with $\mathcal{U}(t)$
and make the transformation 
\begin{equation}
\tilde{U}(t)=\mathcal{U}(t)U(t).\label{eq:transformation}
\end{equation}
Then, the Schr\"odinger equation becomes 
\begin{equation}
\text{i}\partial_{t}\tilde{U}(t)=\tilde{H}_{\text{tot}}(t)\tilde{U}(t),
\end{equation}
where 
\begin{equation}
\tilde{H}_{\text{tot}}(t)=\mathcal{U}(t)H_{\text{tot}}(t)\mathcal{U}^{\dagger}(t)-\text{i}\mathcal{U}(t)\partial_{t}\mathcal{U}^{\dagger}(t).\label{eq:Htot-tilde}
\end{equation}
Using Eqs.~(\ref{eq:transformation},~\ref{eq:Htot-tilde}), the
generator can be rewritten as 
\begin{align}
G_{t_{f}} & =\int_{0}^{t_{f}}\tilde{U}^{\dagger}(\tau)\mathcal{U}(\tau)\partial_{\lambda}H_{\lambda}\mathcal{U}^{\dagger}(\tau)\tilde{U}(\tau)d\tau\nonumber \\
 & =\int_{0}^{t_{f}}\tilde{U}^{\dagger}(\tau)\partial_{\lambda}\left[\mathcal{U}(\tau)H_{\text{tot}}(\tau)\mathcal{U}^{\dagger}(\tau)-\text{i}\mathcal{U}(\tau)\partial_{t}\mathcal{U}^{\dagger}(\tau)\right]\tilde{U}(\tau)d\tau\nonumber \\
 & =\int_{0}^{t_{f}}\tilde{U}^{\dagger}(\tau)\partial_{\lambda}\tilde{H}_{\text{tot}}(\tau)\tilde{U}(\tau)d\tau,
\end{align}
where we have used the important fact that $\mathcal{U}(t)$ is independent
of the parameter $\lambda$. Eq.~(\ref{eq:no-leakage}) in the interaction
frame can be rewritten as 
\begin{equation}
\text{Tr}\left\{ \mathcal{X}_{i}\partial_{\lambda}[\Delta\rho(\tau)]\right\} =\text{Tr}\left\{ \tilde{\mathcal{X}}_{i}(\tau)\partial_{\lambda}[\Delta\tilde{\rho}(\tau)]\right\},
\end{equation}
where 
\begin{align}
\tilde{\mathcal{X}}_{i}(\tau) & \equiv\mathcal{U}(\tau)\mathcal{X}_{i}\mathcal{U}^{\dagger}(\tau),\\
\Delta\tilde{\rho}(\tau) & \equiv\tilde{U}(\tau)\Delta\rho\tilde{U}^{\dagger}(\tau).
\end{align}

\section{\label{sec:Commutators}Commutators in many-spin systems}

In this section, we discuss several properties for many-spin systems.
We define the $k$-body Hermitian basis operator 
\begin{equation}
\mathcal{X}_{k}^{\alpha_{k}}=\sigma_{i_{1}}^{\alpha_{i_{1}}}\sigma_{i_{2}}^{\alpha_{i_{2}}}\cdots\sigma_{i_{k}}^{\alpha_{i_{k}}},
\end{equation}
where the distinct indices $i_{1}<i_{2},\,\cdots<i_{k}$ take values
in $[1,\,n]$ , $\alpha_{i_{k}}\in\{x,\,y,\,z\}$ and $\alpha_{k}\equiv(\alpha_{i_{1}},\,\cdots,\alpha_{i_{k}})$.
The $0-$body operator $\mathcal{X}_{0}$ is defined to be a constant.
We first discuss a lemma concerning the general property of the commutator
between a $k-$body operator and a $l-$body operator. We denote by $p$
 the number of identical subscripts for two many-body operators $\mathcal{X}_{k}^{\alpha_{k}}$ and $\mathcal{X}_{l}^{\alpha_{l}}$.
For $p=0$,  we know $[\mathcal{X}_{k}^{\alpha_{k}},\,\mathcal{X}_{l}^{\alpha_{l}}]=0$.
For positive $p$, we have the following lemma:
\begin{lem}
\label{lem:commutator}Consider two many-body spin operators, where $\mathcal{X}_{k}=\sigma_{i_{1}}^{\alpha_{i_{1}}}\sigma_{i_{2}}^{\alpha_{i_{2}}}\cdots\sigma_{i_{k}}^{\alpha_{i_{k}}}$
and $\mathcal{X}_{l}=\sigma_{j_{1}}^{\alpha_{j_{1}}}\sigma_{j_{2}}^{\alpha_{j_{2}}}\cdots\sigma_{j_{l}}^{\alpha_{j_{l}}}$.
Without loss of generality,  assume $i_{1}=j_{1},\cdots,\,$
$i_{p}=j_{p}$ with the remaining subscripts being distinct and $\alpha_{i_{1}}=\alpha_{j_{1}},\,\cdots,\,$
$\alpha_{i_{q}}=\alpha_{j_{q}}$ with the remaining $p-q$ pairs of
superscripts being distinct, i.e., $\alpha_{i_{q+1}}\neq\alpha_{j_{q+1}},\,\cdots,\,$
$\alpha_{i_{p}}\neq\alpha_{j_{p}}$. Then, 
\begin{equation}
[\mathcal{X}_{k},\,\mathcal{X}_{l}]=\begin{cases}
0 & p-q=\text{even}\\
\mathcal{X}_{k+l-(p-q)} & p-q=\text{odd}
\end{cases},\label{eq:commutator}
\end{equation}
where $p\in[0,\,\min\{k,\,l\}]$ and $q\in[0,\,p]$.
\end{lem}

\begin{proof}
It is straightforward calculate 
\begin{equation}
[\mathcal{X}_{k},\,\mathcal{X}_{l}]=\left[\prod_{r=1}^{p}\left(\sigma_{i_{r}}^{\alpha_{i_{r}}}\sigma_{j_{r}}^{\alpha_{j_{r}}}\right)-\prod_{r=1}^{p}\left(\sigma_{j_{r}}^{\alpha_{j_{r}}}\sigma_{i_{r}}^{\alpha_{i_{r}}}\right)\right]\sigma_{i_{p+1}}^{\alpha_{i_{p+1}}}\dots\sigma_{i_{k}}^{\alpha_{i_{k}}}\sigma_{i_{p+1}}^{\alpha_{i_{p+1}}}\dots\sigma_{i_{l}}^{\alpha_{i_{l}}}.
\end{equation}
With the identity $\sigma_{i_{r}}^{\alpha_{i_{r}}}\sigma_{i_{r}}^{\alpha_{j_{r}}}=\delta^{\alpha_{i_{r}}\alpha_{j_{r}}}+\text{i}\epsilon_{\alpha_{i_{r}}\alpha_{j_{r}}\beta}\sigma_{i_{r}}^{\beta}$, where $\epsilon_{\alpha_{i_{r}}\alpha_{j_{r}}\beta}$ is the Levi-Civita symbol,
we find 
\begin{equation}
\sigma_{i_{r}}^{\alpha_{i_{r}}}\sigma_{i_{r}}^{\alpha_{j_{r}}}=\begin{cases}
1 & \alpha_{i_{r}}=\alpha_{j_{r}}\\
-\sigma_{i_{r}}^{\alpha_{j_{r}}}\sigma_{i_{r}}^{\alpha_{i_{r}}}\sim\sigma_{i_{r}}^{\beta_{i_{r}}} & a_{i_{r}}\neq a_{j_{r}}
\end{cases}.
\end{equation}
Therefore,
\begin{align}
\prod_{r=1}^{p}\left(\sigma_{i_{r}}^{\alpha_{i_{r}}}\sigma_{j_{r}}^{\alpha_{j_{r}}}\right) & =\prod_{r=q+1}^{p}\left(\sigma_{i_{r}}^{a_{i_{r}}}\sigma_{j_{r}}^{\alpha_{j_{r}}}\right)=\prod_{r=q+1}^{p}\left(\sigma_{i_{r}}^{\alpha_{j_{r}}}\sigma_{i_{r}}^{\alpha_{i_{r}}}\right)(-1)^{p-q},
\end{align}

\begin{equation}
\prod_{r=1}^{p}\left(\sigma_{j_{r}}^{\alpha_{j_{r}}}\sigma_{i_{r}}^{\alpha_{i_{r}}}\right)=\prod_{r=q+1}^{p}\left(\sigma_{j_{r}}^{\alpha_{j_{r}}}\sigma_{i_{r}}^{\alpha_{i_{r}}}\right),
\end{equation}
and we conclude that
\begin{equation}
\prod_{r=1}^{p}\left(\sigma_{i_{r}}^{\alpha_{i_{r}}}\sigma_{j_{r}}^{a_{j_{r}}}\right)-\prod_{r=1}^{p}\left(\sigma_{j_{r}}^{\alpha_{j_{r}}}\sigma_{i_{r}}^{\alpha_{i_{r}}}\right)=\begin{cases}
0 & p-q=\text{even}\\
\prod_{r=q+1}^{p}\sigma_{i_{r}}^{\beta_{i_{r}}} & p-q=\text{odd}
\end{cases},\label{eq:commutator-general}
\end{equation}
which completes the proof.
\end{proof}
We next consider the following commutator

\begin{align}
 & \big[\sum_{i=1}^{n}z_{i}^{\alpha\beta}\sigma_{i}^{\alpha}\sigma_{i+1}^{\beta},\,\sum_{j=1}^{n}w_{j}^{\gamma\delta*}\sigma_{j}^{\gamma}\sigma_{j+1}^{\delta}\big]=\sum_{i=1}^{n}\sum_{j=1}^{n}z_{j\alpha\beta}w_{j}^{\gamma\delta*}[\sigma_{i}^{\alpha}\sigma_{i+1}^{\beta},\,\sigma_{j}^{\gamma}\sigma_{j+1}^{\delta}]\nonumber \\
= & \sum_{i=1}^{n}z_{i}^{\alpha\beta}w_{i}^{\gamma\delta*}[\sigma_{i}^{\alpha}\sigma_{i+1}^{\beta},\,\sigma_{i}^{\gamma}\sigma_{i+1}^{\delta}]+\sum_{i=1}^{n}z_{i}^{\alpha\beta}w_{i+1}^{\gamma\delta*}\sigma_{i}^{\alpha}[\sigma_{i+1}^{\beta},\,\sigma_{i+1}^{\gamma}]\sigma_{i+2}^{\delta}\nonumber \\
+ & \sum_{i=1}^{n}z_{i+1}^{\alpha\beta}w_{i}^{\gamma\delta*}\sigma_{i}^{\gamma}[\sigma_{i+1}^{\alpha},\,\,\sigma_{i+1}^{\delta}]\sigma_{i+2}^{\beta},\label{eq:cross-comm}
\end{align}
where $z_{i}^{\alpha\beta}$ and $w_{j\gamma\delta}$ are complex numbers
and we have used 
\begin{equation}
\sum_{i=1}^{n}z_{i}^{\alpha\beta}w_{i-1}^{\gamma\delta}[\sigma_{i}^{\alpha}\sigma_{i+1}^{\beta},\,\sigma_{i-1}^{\gamma}\sigma_{i}^{\delta}]=\sum_{i=1}^{n}z_{i+1}^{\alpha\beta}w_{i}^{\gamma\delta}\sigma_{i}^{\gamma}[\sigma_{i+1}^{\alpha},\,\,\sigma_{i+1}^{\delta}]\sigma_{i+2}^{\beta},
\end{equation}
thanks to the periodic boundary condition. Using Eq.~(\ref{eq:full-comm}),
it is straightforward to obtain
\begin{align}
 & \big[\sum_{i=1}^{n}z_{i}^{\alpha\beta}\sigma_{i}^{\alpha}\sigma_{i+1}^{\beta},\,\sum_{j=1}^{n}z_{j}^{\alpha\beta*}\sigma_{j}^{\alpha}\sigma_{j+1}^{\beta}\big]\nonumber \\
= & \sum_{i=1}^{n}z_{i}^{\alpha\beta}z_{i+1}^{\alpha\beta*}\sigma_{i}^{\alpha}[\sigma_{i+1}^{\beta},\,\sigma_{i+1}^{\alpha}]\sigma_{i+2}^{\beta}+\sum_{i=1}^{n}z_{i+1}^{\alpha\beta}z_{i}^{\alpha\beta*}\sigma_{i}^{\alpha}[\sigma_{i+1}^{\alpha},\,\sigma_{i+1}^{\beta}]\sigma_{i+2}^{\beta}\nonumber \\
= & -2\text{i}\epsilon_{\alpha\beta\gamma}\sum_{i=1}^{n}(z_{i}^{\alpha\beta}z_{i+1}^{\alpha\beta*}-z_{i}^{\alpha\beta*}z_{i+1}^{\alpha\beta})\sigma_{i}^{\alpha}\sigma_{i+1}^{\gamma}\sigma_{i+2}^{\beta}\nonumber \\
= & 4\epsilon_{\alpha\beta\gamma}\sum_{i=1}^{n}\text{Im}(z_{i}^{\alpha\beta}z_{i+1}^{\alpha\beta*})\sigma_{i}^{\alpha}\sigma_{i+1}^{\gamma}\sigma_{i+2}^{\beta}.\label{eq:self-comm}
\end{align}
With Eqs.~(\ref{eq:cross-comm},~\ref{eq:self-comm}), one finds
\begin{align}
 & [\sum_{i=1}^{n}z_{i}^{\alpha\beta}\sigma_{i}^{\alpha}\sigma_{i+1}^{\beta}+\sum_{j=1}^{n}w_{j}^{\gamma\delta}\sigma_{j}^{\gamma}\sigma_{j+1}^{\delta},\,\sum_{i=1}^{n}z_{i}^{\alpha\beta*}\sigma_{i}^{\alpha}\sigma_{i+1}^{\beta}+\sum_{j=1}^{n}w_{j}^{\gamma\delta*}\sigma_{j}^{\gamma}\sigma_{j+1}^{\delta}]\nonumber \\
= & 4\epsilon_{\alpha\beta\mu}\sum_{i=1}^{n}\text{Im}(z_{i}^{\alpha\beta}z_{i+1}^{\alpha\beta*})\sigma_{i}^{\alpha}\sigma_{i+1}^{\mu}\sigma_{i+2}^{\beta}+4\epsilon_{\gamma\delta\lambda}\sum_{i=1}^{n}\text{Im}(w_{i}^{\alpha\beta}w_{i+1}^{\alpha\beta*})\sigma_{i}^{\gamma}\sigma_{i+1}^{\lambda}\sigma_{i+2}^{\delta}\nonumber \\
- & 4\epsilon_{\beta\gamma\mu}\sum_{i=1}^{n}\text{Im}(z_{i}^{\alpha\beta}w_{i+1}^{\gamma\delta*})\sigma_{i}^{\alpha}\sigma_{i+1}^{\mu}\sigma_{i+2}^{\delta}-4\epsilon_{\alpha\delta\lambda}\sum_{i=1}^{n}\text{Im}(z_{i+1}^{\alpha\beta}w_{i}^{\gamma\delta*})\sigma_{i}^{\gamma}\sigma_{i+1}^{\lambda}\sigma_{i+2}^{\beta}+\mathcal{X}_{\leq2}^{\alpha\beta\gamma\delta},\label{eq:full-comm}
\end{align}
where we have used the fact that $[A+B,\,A^{\dagger}+B^{\dagger}]=[A,\,A^{\dagger}]+[B,\,B^{\dagger}]+([A,\,B^{\dagger}]+\text{h.c.})$
and that
\begin{equation}
\mathcal{X}_{\leq2}^{\alpha\beta\gamma\delta}\equiv2\text{i}\sum_{i=1}^{n}\text{Im}(z_{i}^{\alpha\beta}w_{i}^{\gamma\delta*})[\sigma_{i}^{\alpha}\sigma_{i+1}^{\beta},\,\sigma_{i}^{\gamma}\sigma_{i+1}^{\delta}]
\end{equation}
is at most two-body. 

To cancel the three-body interaction $\sigma_{i}^{x}\sigma_{i+1}^{x}\sigma_{i+2}^{x}$
in Eq.~(\ref{eq:spin-chain}) in the main text, we would like the
second term of Eq.~(\ref{eq:cross-comm}) to be the three body operator.
One may take $\alpha=x$, $\beta=y$, $\gamma=z$ and $\delta=x$
in Eq.~(\ref{eq:cross-comm}). This results in 
\begin{equation}
[\sum_{i=1}^{n}z_{i}^{xy}\sigma_{i}^{x}\sigma_{i+1}^{y},\,\sum_{j=1}^{n}w_{j}^{zx*}\sigma_{j}^{z}\sigma_{j+1}^{x}]=2\text{i}\sum_{i=1}^{n}z_{i}^{xy}w_{i+1}^{zx*}\sigma_{i}^{x}\sigma_{i+1}^{x}\sigma_{i+2}^{x},\label{eq:three-body}
\end{equation}
where have used the fact that 
\begin{align}
[\sigma_{i}^{x}\sigma_{i+1}^{y},\,\sigma_{i}^{z}\sigma_{i+1}^{x}] & =0,\\
\sigma_{i}^{z}[\sigma_{i+1}^{x},\,\sigma_{i+2}^{x}]\sigma_{i+2}^{y} & =0,
\end{align}
according to Eq.~(\ref{eq:commutator-general}). Thus, if we take
\begin{equation}
H_{\text{c},\,l}=\sum_{i=1}^{n}\left(c_{li}^{xy}\sigma_{i}^{x}\sigma_{i+1}^{y}+c_{li}^{zx}\sigma_{i}^{z}\sigma_{i+1}^{x}\right)
\end{equation}
 discussed in the main text, according to Eq.~(\ref{eq:full-comm}),
one can easily find 
\begin{align}
[H_{\text{c},\,l},\,H_{\text{c},\,-l}] & =4\sum_{i=1}^{n}\text{Im}(c_{li}^{xy}c_{li+1}^{xy*})\sigma_{i}^{x}\sigma_{i+1}^{z}\sigma_{i+2}^{y}\nonumber \\
 & +4\sum_{i=1}^{n}\text{Im}(c_{li}^{zx}c_{li+1}^{zx*})\sigma_{i}^{z}\sigma_{i+1}^{y}\sigma_{i+2}^{x}\nonumber \\
 & -4\sum_{i=1}^{n}\text{Im}(c_{li}^{xy}c_{li+1}^{zx*})\sigma_{i}^{x}\sigma_{i+1}^{x}\sigma_{i+2}^{x}.\label{eq:Hc-comm}
\end{align}
To cancel the three-body interactions in Eq.~(\ref{eq:spin-chain}),
we would like the coefficients in the first and second terms on the r.h.s.
of Eq.~(\ref{eq:Hc-comm}) to vanish, while the coefficient in last term
on the r.h.s. does not vanish. To determine the conditions for this to be the case,
we first introduce the following lemma 
\begin{lem}
\label{lem:coeff}Given a set of complex number $\{z_{i}\}_{i=1}^{n}$
and periodic boundary condition $z_{n+1}=z_{1}$, the condition 
\begin{equation}
\text{Im}(z_{i}z_{i+1}^{*})=0,\,\forall i\label{eq:Imzz}
\end{equation}
is satisfied if and only if 
\begin{equation}
\frac{\text{Re}(z_{i})}{\text{Im}(z_{i})}=\alpha,\label{eq:z-form}
\end{equation}
with $\alpha$ independent of the index $i$.
\end{lem}

\begin{proof}
Assuming Eq.~(\ref{eq:z-form}), Eq.~(\ref{eq:Imzz}) is obvious.
The other direction can be proved by separating $z_{i}$ into real
and imaginary part, i.e., $z_{i}=u_{i}+\text{i}v_{i}$. Eq.~(\ref{eq:Imzz})
leads to $u_{i}/v_{i}=u_{i+1}/v_{i+1}$, $\forall i$. We may thus set $u_{i}/v_{i}=\alpha$, where $\alpha$ is
real constant and independent of the index $i$. This concludes the
proof. 
\end{proof}
Therefore, according Lemma~\ref{lem:coeff}, in order to obtain vanishing
coefficients in the first and second terms on the r.h.s. of Eq.~(\ref{eq:Hc-comm}),
the coefficients should take the following form 
\begin{align}
c_{li}^{xy} & =(\alpha_{l}^{xy}+\text{i})v_{li}^{xy},\\
c_{li}^{zx} & =(\alpha_{l}^{zx}+\text{i})v_{li}^{zx}.
\end{align}
As a result, we obtain
\begin{equation}
\text{Im}(c_{li}^{xy}c_{li+1}^{zx*})=(\alpha_{l}^{xy}-\alpha_{l}^{zx})v_{li}^{xy}v_{li+1}^{zx},
\end{equation}
and 
\begin{equation}
[H_{\text{c},\,l},\,H_{\text{c},\,-l}]=4(\alpha_{l}^{zx}-\alpha_{l}^{xy})\sum_{i=1}^{n}v_{li}^{xy}v_{li+1}^{zx}\sigma_{i}^{x}\sigma_{i+1}^{x}\sigma_{i+2}^{x}.
\end{equation}
Thus, it is necessary that $\alpha_{l}^{zx}\neq\alpha_{l}^{xy}$ to have
non-vanishing $[H_{\text{c},\,l},\,H_{\text{c},\,-l}]$. The choice
in the main text corresponds to the case where $\alpha_{l}^{xy}=\infty$
and $\alpha_{l}^{zx}=0$, so that $c_{li}^{xy}$ is real and $c_{li}^{zx}$
is purely imaginary. 

\section{\label{sec:More-general-driving}More general driving protocols}

In this section, we generalize the mechanism of the cancellation of
the the $\sigma_{x}$-type interactions discussed in the main text
to arbitrary type of three-body interactions, $J_{\kappa\mu\lambda}/2\sum_{i=1}^{n}\sigma_{i}^{\kappa}\sigma_{i+1}^{\mu}\sigma_{i+2}^{\lambda}$.
We employ Eqs.~(\ref{eq:cross-comm},~\ref{eq:self-comm},~\ref{eq:full-comm})
and Lemma~\ref{lem:coeff} to construct the high frequency drive
$H^{(d)}(t)$ in the main text. Note that when assuming two-body controls
are accessible, any two-body interaction out of the commuting $H_{\text{c},\,l}$
and $H_{\text{c},\,-l}$ $(l\ge1)$ can be cancel by adding control
to the static control $H_{\text{c},\,0}$. Therefore, one should not
concern about the last term $\mathcal{X}_{\leq2}^{\alpha\beta\gamma\delta}$
on the r.h.s. of Eq.~(\ref{eq:full-comm}) , which is at most two-body
interaction according to Lemma~\ref{lem:commutator}. Now let us
discuss case by case according to the index degeneracy. 
\begin{enumerate}
\item For three-body interaction $J_{\kappa\kappa\kappa}/2\sum_{i=1}^{n}\sigma_{i}^{\kappa}\sigma_{i+1}^{\kappa}\sigma_{i+2}^{\kappa}$,
we take $\alpha=\delta=\kappa$, $\beta=\lambda$ and $\gamma=\mu$
in Eq.~(\ref{eq:cross-comm}), where $\kappa,\,\mu$ and $\lambda$
are distinct and $[\sigma_{i}^{\lambda},\,\sigma_{i}^{\mu}]=2\text{i}\sigma_{i}^{\kappa}$,
so that the second term on the r.h.s generates $\sigma_{i}^{\kappa}\sigma_{i+1}^{\kappa}\sigma_{i+2}^{\kappa}$
. Then Eq.~(\ref{eq:cross-comm}) becomes 
\begin{equation}
[\sum_{i=1}^{n}z_{i}^{\kappa\lambda}\sigma_{i}^{\kappa}\sigma_{i+1}^{\lambda},\,\sum_{j=1}^{n}w_{j}^{\mu\kappa*}\sigma_{j}^{\mu}\sigma_{j+1}^{\kappa}]=2\text{i}\sum_{i=1}^{n}z_{i}^{\kappa\lambda}w_{i+1}^{\mu\kappa*}\sigma_{i}^{\kappa}\sigma_{i+1}^{\kappa}\sigma_{i+2}^{\kappa}
\end{equation}
We construct
\begin{equation}
H_{\text{c},\,l}=\sum_{i=1}^{n}\left(c_{li}^{\kappa\lambda}\sigma_{i}^{\kappa}\sigma_{i+1}^{\lambda}+c_{li}^{\mu\kappa}\sigma_{i}^{\mu}\sigma_{i+1}^{\kappa}\right)
\end{equation}
and assume the coefficients $c_{li}^{\kappa\lambda}$ and $c_{li}^{\mu\kappa}$
satisfy Eq.~(\ref{eq:z-form}).Then according to Eq.~(\ref{eq:full-comm})
and Lemma~\ref{lem:coeff}, one can easily find 
\begin{equation}
[H_{\text{c},\,l},\,H_{\text{c},\,-l}]=-4\sum_{i=1}^{n}\text{Im}(c_{li}^{\kappa\lambda}c_{li+1}^{\mu\kappa*})\sigma_{i}^{\kappa}\sigma_{i+1}^{\kappa}\sigma_{i+2}^{\kappa}.
\end{equation}
Assuming the coefficients are independent of $i$, the three-body
interaction $J_{\kappa\kappa\kappa}\sum_{i=1}^{n}\sigma_{i}^{\kappa}\sigma_{i+1}^{\kappa}\sigma_{i+2}^{\kappa}$
can be cancelled if 
\begin{equation}
\frac{J_{\kappa\kappa\kappa}}{2}=\sum_{l=1}^{\infty}\frac{1}{l}\text{Im}(c_{l}^{\kappa\lambda}c_{l}^{\mu\kappa*}).\label{eq:FMC-3kappa}
\end{equation}
\item For three-body interaction $J_{\kappa\mu\kappa}/2\sum_{i=1}^{n}\sigma_{i}^{\kappa}\sigma_{i+1}^{\mu}\sigma_{i+2}^{\kappa}$,
where $\mu\neq\kappa$, we take $\alpha=\delta=\kappa$, $\beta=\lambda$
and $\gamma=\kappa$ in Eq.~(\ref{eq:cross-comm}), where $\kappa,\,\mu$
and $\lambda$ are distinct and $[\sigma_{i}^{\lambda},\,\sigma_{i}^{\kappa}]=2\text{i}\epsilon_{\lambda\kappa\mu}\sigma_{i}^{\mu}$.
Eq.~(\ref{eq:cross-comm}) becomes
\begin{equation}
[\sum_{i=1}^{n}z_{i}^{\kappa\lambda}\sigma_{i}^{\kappa}\sigma_{i+1}^{\lambda},\,\sum_{j=1}^{n}w_{j}^{\kappa\kappa*}\sigma_{j}^{\kappa}\sigma_{j+1}^{\kappa}]=\mathcal{X}_{2}^{\kappa\lambda\kappa\kappa}+2\text{i}\epsilon_{\lambda\kappa\mu}\sum_{i=1}^{n}z_{i}^{\kappa\lambda}w_{i+1}^{\kappa\kappa*}\sigma_{i}^{\kappa}\sigma_{i+1}^{\mu}\sigma_{i+2}^{\kappa},
\end{equation}
where the two-body interaction $\mathcal{X}_{2}^{\kappa\lambda\kappa\kappa}$
can be always cancelled through $H_{\text{c},\,0}$. One can then
construct
\begin{equation}
H_{\text{c},\,l}=\sum_{i=1}^{n}\left(c_{li}^{\kappa\lambda}\sigma_{i}^{\kappa}\sigma_{i+1}^{\lambda}+c_{li}^{\kappa\kappa}\sigma_{i}^{\kappa}\sigma_{i+1}^{\kappa}\right).
\end{equation}
According to Eq.~(\ref{eq:full-comm}), one can readily find 
\begin{equation}
[H_{\text{c},\,l},\,H_{\text{c},\,-l}]=4\epsilon_{\lambda\kappa\mu}\sum_{i=1}^{n}\text{Im}(z_{i}^{\kappa\lambda}w_{i+1}^{\kappa\kappa*})\sigma_{i}^{\kappa}\sigma_{i+1}^{\mu}\sigma_{i+2}^{\kappa}+\mathcal{X}_{\leq2}^{\alpha\beta\gamma\delta}
\end{equation}
with a similar AFM condition~(\ref{eq:FMC-3kappa}).
\item For three-body interaction $J_{\kappa\mu\lambda}/2\sum_{i=1}^{n}\sigma_{i}^{\kappa}\sigma_{i+1}^{\mu}\sigma_{i+2}^{\lambda}$,
where $\kappa,\,\mu$ and $\lambda$ are distinct, we take $\alpha=\gamma=\kappa$,
$\beta=\delta=\lambda$ in Eq.~(\ref{eq:cross-comm}) and obtain
\begin{align}
[\sum_{i=1}^{n}z_{i}^{\kappa\lambda}\sigma_{i}^{\kappa}\sigma_{i+1}^{\lambda},\,\sum_{j=1}^{n}w_{j}^{\mu\kappa*}\sigma_{j}^{\kappa}\sigma_{j+1}^{\lambda}] & =2\text{i}\epsilon_{\lambda\kappa\mu}\sum_{i=1}^{n}z_{i}^{\kappa\lambda}w_{i+1}^{\kappa\lambda*}\sigma_{i}^{\kappa}\sigma_{i+1}^{\mu}\sigma_{i+2}^{\lambda}+2\text{i}\epsilon_{\kappa\lambda\mu}\sum_{i=1}^{n}z_{i+1}^{\kappa\lambda}w_{i}^{\kappa\lambda*}\sigma_{i}^{\kappa}\sigma_{i+1}^{\mu}\sigma_{i+2}^{\lambda}\nonumber \\
 & =-2\text{i}\epsilon_{\kappa\lambda\mu}\sum_{i=1}^{n}[z_{i}^{\kappa\lambda}w_{i+1\kappa\lambda}^{*}-z_{i+1}^{\kappa\lambda}w_{i}^{\kappa\lambda*}]\sigma_{i}^{\kappa}\sigma_{i+1}^{\mu}\sigma_{i+2}^{\lambda},
\end{align}
which suggests us to use the two-body interaction $\sigma_{i}^{\kappa}\sigma_{i+1}^{\lambda}$
as the only drive. Indeed, if we take
\begin{equation}
H_{\text{c},\,l}=\sum_{i=1}^{n}c_{li}^{\kappa\lambda}\sigma_{i}^{\kappa}\sigma_{i+1}^{\lambda},
\end{equation}
for $l\ge1$ then according to Eq.~(\ref{eq:self-comm}), we find
\begin{equation}
[H_{\text{c},\,l},\,H_{\text{c},\,-l}]=4\epsilon_{\kappa\mu\lambda}\sum_{i=1}^{n}\text{Im}(c_{li}^{\kappa\lambda}c_{li+1}^{\kappa\lambda*})\sigma_{i}^{\kappa}\sigma_{i+1}^{\mu}\sigma_{i+2}^{\lambda}.
\end{equation}
According to Lemma~\ref{lem:coeff}, as long as the ratio $\text{Re}(c_{li}^{\kappa\lambda})/\text{Im}(c_{li}^{\kappa\lambda})\equiv\alpha_{li}^{\kappa\lambda}$
varies across the chain, the three-body interaction is non-zero. Denoting
\begin{equation}
c_{li}^{\kappa\lambda}=(\alpha_{li}^{\kappa\lambda}+\text{i})v_{li}^{\kappa\lambda},
\end{equation}
we find 
\begin{equation}
\text{Im}(c_{li}^{\kappa\lambda}c_{li+1}^{\kappa\lambda*})=(\alpha_{li+1}^{\kappa\lambda}-\alpha_{li}^{\kappa\lambda})v_{li+1}^{\kappa\lambda}v_{li}^{\kappa\lambda}.
\end{equation}
A simple example could be $\alpha_{li+1}^{\kappa\lambda}-\alpha_{li}^{\kappa\lambda}=\Delta\alpha_{l}^{\kappa\lambda}$
and $v_{li}^{\kappa\lambda}=v_{l}^{\kappa\lambda}$ for $i=1,\,\cdots n-1$.
Therefore $\alpha_{li+1}^{\kappa\lambda}=\alpha_{l1}^{\kappa\lambda}+i\Delta\alpha_{l}^{\kappa\lambda}$
and 
\begin{equation}
\text{Im}(c_{li}^{\kappa\lambda}c_{li+1}^{\kappa\lambda*})=\Delta\alpha_{l}^{\kappa\lambda}[v_{l}^{\kappa\lambda}]^2,
\end{equation}
for $i=1,\,\cdots n-2$. Now we shall choose $\alpha_{ln}^{\kappa\lambda}$
and $v_{ln,\,\kappa\lambda}$ such that 
\begin{align}
\text{Im}(c_{ln-1,\,\kappa\lambda}c_{ln,\,\kappa\lambda}^{*}) & =\left[\alpha_{ln,\,\kappa\lambda}-\alpha_{l1}^{\kappa\lambda}-(n-2)\Delta\alpha_{l}^{\kappa\lambda}\right]v_{ln,\,\kappa\lambda}v_{l}^{\kappa\lambda}=\Delta\alpha_{l}^{\kappa\lambda}[v_{l}^{\kappa\lambda}]^2,\\
\text{Im}(c_{ln,\,\kappa\lambda}c_{l1,\,\kappa\lambda}^{*}) & =\left[\alpha_{l1}^{\kappa\lambda}-\alpha_{ln,\,\kappa\lambda}\right]v_{l}^{\kappa\lambda}v_{ln,\,\kappa\lambda}=\Delta\alpha_{l}^{\kappa\lambda}[v_{l}^{\kappa\lambda}]^2,
\end{align}
from which one finds 
\begin{align}
\alpha_{ln}^{\kappa\lambda} & =\alpha_{1}-\left(1-\frac{n}{2}\right)v_{l}^{\kappa\lambda},\\
v_{ln,\,\kappa\lambda} & =\frac{2\Delta\alpha_{l}^{\kappa\lambda}}{2-n}.
\end{align}
So the AFM condition become 
\begin{equation}
4\epsilon_{\kappa\mu\lambda}\sum_{l=1}^{\infty}\frac{\Delta\alpha_{l}^{\kappa\lambda}[v_{l}^{\kappa\lambda}]^2}{l}+\frac{J_{\kappa\mu\lambda}}{2}=0.
\end{equation}
\item Finally, it turns out, up to the first-order of the high frequency
expansion, there is no simple scheme to cancel the interaction of
the type $\sigma_{i}^{\kappa}\sigma_{i+1}^{\kappa}\sigma_{i+2}^{\lambda}$.
However the combination $1/2\sum_{i=1}^{n}(J_{\kappa\kappa\lambda}\sigma_{i}^{\kappa}\sigma_{i+1}^{\kappa}+J_{\mu\mu\lambda}\sigma_{i}^{\mu}\sigma_{i+1}^{\mu})\sigma_{i+2}^{\lambda}$,
where $\kappa,\,\mu$ and $\lambda$ are distinct and $J_{\kappa\kappa\lambda}=J_{\mu\mu\lambda}$,
, can be cancel in the following simple way. Taking $\alpha=\kappa$,
$\beta=\lambda$, $\gamma=\mu$ and $\delta=\lambda$, one finds 
\begin{equation}
[\sum_{i=1}^{n}z_{i}^{\kappa\lambda}\sigma_{i}^{\kappa}\sigma_{i+1}^{\lambda},\,\sum_{j=1}^{n}w_{j}^{\mu\lambda*}\sigma_{j}^{\mu}\sigma_{j+1}^{\lambda}]=2\text{i}\epsilon_{\lambda\mu\kappa}\sum_{i=1}^{n}(z_{i}^{\kappa\lambda}w_{i+1}^{\mu\lambda*}\sigma_{i}^{\kappa}\sigma_{i+1}^{\kappa}\sigma_{i+2}^{\lambda}+z_{i+1}^{\kappa\lambda}w_{i}^{\mu\lambda*}\sigma_{i}^{\mu}\sigma_{i+1}^{\mu}\sigma_{i+2}^{\lambda})+\mathcal{X}_{2}^{\kappa\lambda\mu\lambda}.
\end{equation}
So if we take 
\begin{equation}
H_{\text{c},\,l}=\sum_{i=1}^{n}\left(c_{li}^{\kappa\lambda}\sigma_{i}^{\kappa}\sigma_{i+1}^{\lambda}+c_{li}^{\mu\lambda}\sigma_{i}^{\mu}\sigma_{i+1}^{\lambda}\right),
\end{equation}
and assume the coefficients $c_{li}^{\kappa\lambda}$ and $c_{li}^{\mu\lambda}$
satisfy Eq.~(\ref{eq:z-form}), the according to Eq.~(\ref{eq:full-comm}),
one find 
\begin{equation}
[H_{\text{c},\,l},\,H_{\text{c},-l}]=-4\epsilon_{\lambda\mu\kappa}\sum_{i=1}^{n}\sum_{i=1}^{n}[\text{Im}(c_{li}^{\kappa\lambda}c_{li+1}^{\mu\lambda*})\sigma_{i}^{\kappa}\sigma_{i+1}^{\kappa}\sigma_{i+2}^{\lambda}+\text{Im}(c_{li+1}^{\kappa\lambda}c_{li}^{\mu\lambda*})\sigma_{i}^{\mu}\sigma_{i+1}^{\mu}\sigma_{i+2}^{\lambda}]+\mathcal{X}_{2}^{\kappa\lambda\mu\lambda}
\end{equation}
Assuming $c_{li}^{\kappa\lambda}$ and $c_{li}^{\mu\lambda}$ are homogeneous,
since satisfy Eq.~(\ref{eq:z-form}), they can take the form $c_{li}^{\kappa\lambda}=(\alpha_{l}^{\kappa\lambda}+\text{i})u_{li}^{\kappa\lambda}$
and $c_{li}^{\mu\lambda}=(\alpha_{l}^{\mu\lambda}+\text{i})v_{li}^{\mu\lambda}$
respectively. Therefore
\begin{align}
\text{Im}(c_{li}^{\kappa\lambda}c_{li+1}^{\mu\lambda*}) & =(\alpha_{l}^{\mu\lambda}-\alpha_{l}^{\kappa\lambda})u_{li}^{\kappa\lambda}v_{li+1}^{\mu\lambda},\\
\text{Im}(c_{li+1}^{\kappa\lambda}c_{li}^{\mu\lambda*}) & =(\alpha_{l}^{\mu\lambda}-\alpha_{l}^{\kappa\lambda})u_{li+1}^{\kappa\lambda}v_{li}^{\mu\lambda}.
\end{align}
Then one can take $u_{i}=\delta_{1,\,i}$ and $v_{i}\neq0$ for $i\neq n$
and $v_{n}=0$. Then one can easily show 
\begin{equation}
\text{Im}(c_{li}^{\kappa\lambda}c_{li+1}^{\mu\lambda*})=(\alpha_{l}^{\mu\lambda}-\alpha_{l}^{\kappa\lambda})u_{li}^{\kappa\lambda}v_{li+1}^{\mu\lambda},
\end{equation}
\begin{equation}
4\epsilon_{\lambda\mu\kappa}\sum_{l=1}^{\infty}\frac{(\alpha_{l}^{\mu\lambda}-\alpha_{l}^{\kappa\lambda})u_{l}^{\kappa\lambda}v_{l}^{\mu\lambda}}{l}=\frac{J_{\kappa\kappa\lambda}}{2}.
\end{equation}
Similar construction also holds for the three-body interaction $1/2\sum_{i=1}^{n}\sigma_{i+2}^{\lambda}(J_{\lambda\kappa\kappa}\sigma_{i}^{\kappa}\sigma_{i+1}^{\kappa}+J_{\lambda\mu\mu}\sigma_{i}^{\mu}\sigma_{i+1}^{\mu})$
with $J_{\lambda\kappa\kappa}=J_{\lambda\mu\mu}$, which will not
be discussed here.
\end{enumerate}

\end{document}